%% file: paper.tex
\documentclass[draftcls,a4paper,onecolumn,12pt]{IEEEtran}
\usepackage{graphicx}
\usepackage{amsmath}
\usepackage{amssymb}
\usepackage[colorlinks]{hyperref}
\usepackage{theorem}

\newif\ifambproofs
\newif\ifdraft
\newif\iffigures

\drafttrue
\figurestrue
\ambproofsfalse

\graphicspath{{./}{/home/jungp/TeX/Figures/}}
\theoremstyle{plain}
\newtheorem{mytheorem}{Theorem}
\newtheorem{mycorollary}[mytheorem]{Corollary}
\newtheorem{mylemma}[mytheorem]{Lemma}
\theoremstyle{definition}
\newtheorem{mydefinition}[mytheorem]{Definition}
\newtheorem{myremark}{Remark}
\newenvironment{myproof}{{\it Proof:$\quad$}}{$\quad$\QED \\}
\newenvironment{myproblem}{\\[1em]{\bf Problem: }\it}{\\[1em]}

\input{notations}

\renewcommand{\DotReal}[2]{#1\cdot #2}
\newcommand{\erf}{\text{erf}}

\begin{document}
\title{
  On the Approximate Eigenstructure of Time--Varying Channels 
}
\author{Peter Jung\\
  Fraunhofer German-Sino Lab for Mobile Communications (MCI) and the Heinrich-Hertz Institute\\[.1em]
  \small{jung@hhi.fraunhofer.de}}
\maketitle
\begin{abstract}
   In this article we consider the approximate description of 
   doubly--dispersive channels by its symbol. We focus on channel operators   
   with compactly supported spreading, which are widely used to represent
   fast fading multipath communication channels. The concept of approximate 
   eigenstructure is introduced, which measures the accuracy $E_p$ of the
   approximation of the channel operation as a pure multiplication in a given $\Leb{p}$--norm.   
   Two variants of such an approximate Weyl symbol calculus are studied, which have
   important applications in several models for time--varying mobile channels. Typically, such channels have
   random spreading functions (inverse Weyl transform) defined on a common support $U$ of
   finite non--zero size such
   that approximate eigenstructure has to be measured with respect to certain norms of the
   spreading process. We derive several explicit relations to the size $|U|$ of the support.
   We show that the characterization of
   the ratio of $E_p$ to some $\Leb{q}$--norm of the spreading function is related 
   to weighted norms of ambiguity and Wigner functions.
   We present the connection to localization operators and give new bounds on the ability of localization of 
   ambiguity functions and Wigner functions in $U$. 
   Our analysis generalizes and improves recent results
   for the case $p=2$ and $q=1$. 
\end{abstract}


%
\section{Introduction}
Optimal signaling through linear time--varying (LTV) channels is a challenging
task for future communication systems. For a particular realization 
of the time--varying channel operator the transmitter and receiver
design, which avoids interference
is related to ''eigen--signaling''. Eigen--signaling simplifies much of the
information theoretic treatment of communication in dispersive channels.
However, it is well known that
for an ensemble of channels, which are dispersive in time and frequency
such a joint diagonalization can not be achieved because the eigen--decompositions can 
differ from one to another channel realization.
Several approaches like for example the ''basis expansion
model'' (BEM) \cite{giannakis:bem} and the canonical channel representation \cite{sayeed:jointdiversity}
are proposed to describe eigen--signaling in some approximate sense.
Then a necessary prerequisite is the characterization
of remaining approximation errors.


A typical scenario commonly encountered 
in wireless communication, is signaling through a random  time-varying and frequency selective (doubly--dispersive)
channel, which in general is represented by a pseudo-differential operator $\BH$. The abstract random
channel operating on an input signal $s:\Reals\rightarrow\Complexes$ can be expressed (at least in the weak sense)
in the form of a random kernel, symbol or spreading function.
The signal $r:\Reals\rightarrow\Complexes$
at the time instant $t$ at the output of the time--varying channel is then:
\begin{equation*}
   r(t)=(\BH s)(t)
\end{equation*}
It is a common assumption that knowledge of $\BH$ at the receiver can be obtained 
up to certain accuracy by channel estimation, which will
allow for coherent detection. However, channel knowledge at the transmitter  
simplifies equalization and detection complexity at the receiver and can increase the
link performance.
It can be used to perform a diagonalizing operation (i.e. eigen--signaling)
and allocation of resources in this domain (e.g. power allocation).
We shall call the first part of this description from now on as the eigenstructure of $\BH$.
Signaling through classes of channels having common eigenstructure could be,
in principle, interference--free  and would allow for simple information recovering algorithms
based on the received signal $r(t)$. However, for $\BH$ being random, random eigenstructure has to be 
expected in general such that the design of the transmitter and the receiver
has to be performed jointly for ensembles of channels having different eigenstructures. Nevertheless, 
interference then can not be avoided in the communication chain. For such interference scenarios it
is important to have bounds on the distortion of a particular selected 
signaling scheme. Refer for example to \cite{durisi:wssus:capacity} 
for a recent application in information theory.

Initial results in this field can be found 
in the literature on pseudo-differential operators 
\cite{kohn:pdo:algebra,folland:harmonics:phasespace} where the overall operator
was split up into a main part to be studied and a ''small'' operator to be controlled.
More recent results with direct application to time--varying channels
were obtained by Kozek \cite{kozek:thesis,kozek:eigenstructure}
and Matz \cite{matz:thesis} which resemble the notion of underspread
channels. They investigated the approximate symbol calculus of pseudo-differential
operators in this context and derived bounds for the $\Leb{2}$--norm
of the distortion which follow from the approximate product rule 
in terms of Weyl symbols. We will present more details on this approach 
in Section \ref{subsec:approxeigen:symbolcalculus}.
Controlling this approximation
intimately scales with the ''size'' of the spreading of the contributing 
channel operators. For operators with compactly supported spreading
such a scale is $|U|$ -- the size of the spreading support $U$.
Interestingly this approximation behavior breaks 
down in their framework at a certain critical size. Channels 
below this critical size are called according to their terminology underspread and
otherwise overspread.
However, we found that previous bounds can be improved and generalized
in several directions by considering the problem of approximate eigenstructure
from another perspective, namely investigating directly the $\Leb{p}$--norm $E_p$ 
of the error $\BH s-\lambda r$
for well known choices of $\lambda$. We shall focus on the case where
$\lambda$ is the symbol of the operator $\BH$ and on the important
case where $\lambda$ is the orthogonal distortion which can be understood
as the $\Leb{2}$--minimizer.
We believe that extensions to $p\neq2$ are important when further statistical
properties of the spreading process of the random channel operator are at hand\footnote{%
  We provide further motivation and arguments in Remark \ref{rem:approxeigen:statisticalmodel} at the end of the paper.}.
Our approach will also show the connection to well known fidelity and
localization criteria related to pulse design 
\cite{kozek:thesis,jung:wssuspulseshaping,jung:isit06}. In particular, 
the latter is also related to the notion of localization operators \cite{daubechiez:tflocalization:geophase}.
The underspread property of doubly--dispersive channels occurs also in the context
of channel measurement and identification \cite{kailath:ltvmeasurement}. 
In addition refer to the following recent articles \cite{kozek:identification:bandlimited,Pfander08:opid} for
rigorous treatments of channel identification
based on Gabor (Weyl--Heisenberg) frame theory. The authors connect
the critical time--frequency sampling density immanent in this theory to
the stability of the channel measurement. A relation between these different
notions of underspreadness has to be expected but is beyond  
the scope of this paper.

The paper is organized as follows: In Section \ref{sec:timefrequencyanalysis} we shall give an 
introduction into the 
basics from time--frequency analysis including the Weyl correspondence and
the spreading representation of doubly--dispersive channels. In Section \ref{sec:mainresults} of the paper we shall
consider the problem of approximate eigenstructure for operators with
spreading functions, which are supported on a common set $U$ in the time--frequency
plane having non--zero and finite
Lebesgue measure $|U|$. We present the approach for $E_2$  followed
by a summary of the main results of our analysis on $E_p$. 
The detailed analysis for $E_p$ will be presented in Section \ref{sec:genanalysis}. Finally,
Section  \ref{sec:numerical} contains a numerical verification of our results.

\subsection{Notation and Some Definitions}
We present certain notation and definitions that shall be used through the paper.
For $1\leq p<\infty$ and functions $f:\Reals^n\rightarrow\Complexes$  
the functionals $\lVert f\rVert_p:=\left(\int|f(t)|^pdt\right)^{1/p}$
are then usual notion of $p$--norms
($dt$ is the Lebesgue measure on $\Reals^n$). 
Furthermore for $p=\infty$ is
$\lVert f\rVert_\infty:=\esup |f(t)|$.
If $\lVert f\rVert_p$ is finite $f$ is said to be in $\Leb{p}(\Reals^n)$. 
The inner product $\langle\cdot,\cdot\rangle$
on the Hilbert space $\Leb{2}(\Reals^n)$ is given as
$\langle x,y\rangle:=\int_{\Reals^n} \bar{x}(t)y(t)dt$ where $\bar{x}(t)$ denotes
complex conjugate of $x(t)$.
A particular dense subset of $\Leb{p}(\Reals^n)$ is 
the class of Schwartz functions $\Schwarz(\Reals^n)$ (infinite
differentiable rapidly decreasing functions). 
The notation $p'$ denotes always the dual index of $p$, i.e.
$1/p+1/p'=1$ with $p'=\infty$ if $p=1$ (and the reverse).

%
\section{Time--Frequency Analysis}
\label{sec:timefrequencyanalysis}
%

\subsection{Phase Space Displacements and Ambiguity Functions }
Several physical properties of time--varying channels (like delay and Doppler spread)  are in 
general related
to a time--frequency view on operators $\BH$. 
Time-frequency representations itself are important tools in signal analysis, physics and 
many other scientific areas. Among them are the 
Woodward cross ambiguity function \cite{woodward:probinfradar}
and the Wigner distribution.
Ambiguity functions can be understood as 
inner product representations of time--frequency shift operators.
More generally, a displacement (or shift) operator for
functions $f:\Reals^n\rightarrow\Complexes$ can be
defined as:
\begin{equation}
   (\Shift_{\mu} f)(x):=e^{i2\pi \DotReal{\mu_2}{x}}f(x-\mu_1)
   \label{eq:weyl:shift:shiftoperator}
\end{equation}
where $\mu=(\mu_1,\mu_2)\in\Reals^{2n}$ and $\mu_1,\mu_2\in\Reals^n$. In general
$\Reals^{2n}$ is called phase space.
Later on
we shall focus on $n=1$, where we have that the functions $f$ are signals in time and $\mu$ is a displacement in
time and frequency. Then the phase space is also called \emph{time--frequency plane} and 
the operators $\Shift_\mu$ are \emph{time--frequency shift operators}.
There is an ambiguity as to which displacement should be performed first where
\eqref{eq:weyl:shift:shiftoperator} corresponds to the separation 
$\Shift_\mu=\Shift_{(0,\mu_2)}\Shift_{(\mu_1,0)}$.
However, it is well known that
a generalized view can be achieved by considering so--called $\alpha$-generalized displacements: 
\begin{equation}
   \Shift_{\mu}^{(\alpha)}:=
   \Shift_{(0,\mu_2(\frac{1}{2}+\alpha))}
   \Shift_{(\mu_1,0)}
   \Shift_{(0,\mu_2(\frac{1}{2}-\alpha))}
   =e^{-i2\pi(1/2-\alpha)\asy(\mu,\mu)}\Shift_\mu
   \label{eq:weyl:shift:alphageneralized}
\end{equation}
where $\asy(\mu,\nu)=\DotReal{\mu_1}{\nu_2}$ (inner product on $\Reals^n$)
and then set $\Shift_\mu=\Shift_\mu^{(1/2)}$. Usually
$\alpha$ is called \emph{polarization}.
The operators in \eqref{eq:weyl:shift:alphageneralized} act isometrically on all $\Leb{p}(\Reals^n)$, 
hence are unitary on $\Leb{2}(\Reals^n)$. 
Furthermore, they establish\footnote{up to unitary equivalence}
unitary representations (Schr\"odinger representation) of the Weyl--Heisenberg
group on $\Ltwo(\Reals^n)$ (see for example \cite{folland:harmonics:phasespace}). 
In physics it is common to choose the most symmetric case $\alpha=0$ 
and the operators are usually  called  Weyl operators or
Glauber displacement operators. 
If we define the symplectic form as
$\sympl(\mu,\nu):=\asy(\mu,\nu)-\asy(\nu,\mu)$,
we have the following well known \emph{Weyl commutation relation}:
\begin{equation}
   \begin{split} 
      \Shift_{\mu}^{(\alpha)}\Shift_{\nu}^{(\beta)} =
      e^{-i2\pi \sympl(\mu,\nu)}\Shift_{\nu}^{(\beta)}\Shift_{\mu}^{(\alpha)}
   \end{split} 
   \label{eq:weyl:shift:weylcommrelation}
\end{equation}
for arbitrary  polarizations $\alpha$ and $\beta$.
In this way a generalized (cross) ambiguity function can be defined as:
\begin{equation}
   \Amb^{(\alpha)}_{g\gamma}(\mu)\defeq\langle g,\Shift_\mu^{(\alpha)}\gamma\rangle
   =\int_{\Reals^n}\bar{g}(x+(\frac{1}{2}-\alpha)\mu_1)\gamma(x-(\frac{1}{2}+\alpha)\mu_1)
   e^{i2\pi\DotReal{\mu_2}{x}}dx
   \label{eq:tfanalysis:crossamb}
\end{equation}
The function $\Amb^{(1/2)}_{g\gamma}$ is also known as the 
\emph{Short--time Fourier transform} (sometimes also windowed Fourier transform or
Fourier--Wigner transform) of $g$ 
with respect to a window $\gamma$. This function is continuous for 
$g\in\Schwarz(\Reals^n)$ and $\gamma\in\Schwarz'(\Reals^n)$ (the dual of
$\Schwarz(\Reals^n)$, i.e. the tempered distributions).
Well known relations of these functions, which follow
directly from definition \eqref{eq:tfanalysis:crossamb} are:
\begin{equation}
   |\Amb^{(\alpha)}_{g\gamma}(\mu)|=|\langle g,\Shift_\mu^{(\alpha)}\gamma\rangle|\leq
   \lVert g\rVert_2\lVert\gamma\rVert_2=\lVert\Amb_{g\gamma}^{(\alpha)}\rVert_2
   \label{eq:tfanalysis:crossamb:properties}
\end{equation}
where the right hand side (rhs) is sometimes also called the radar uncertainty principle. For particular 
weight functions $m:\Reals^{2n}\rightarrow\RealsPlus$ the weighted $p$--norms
$\lVert\Amb_{g\gamma}^{(\alpha)}m\rVert_p$ are also called the modulation norms
$\lVert\gamma\rVert_{M_{m}^{p,p}}$ of $\gamma$ with respect to Schwartz 
function $g\in\Schwarz(\Reals^n)$ ($M_m^{p,p}$ is then corresponding modulation space \cite{feichtinger:modspaces}).
Let the symplectic Fourier transform $\sFourier F$ of a function 
$F:\Reals^{2n}\rightarrow\Complexes$ be defined as: 
\begin{equation}
   (\sFourier F)(\mu)=\int_{\Reals^{2n}} e^{-i2\pi\sympl(\nu,\mu)} F(\nu)d\nu
   \label{eq:tfanalysis:symplectic:fourier}
\end{equation}
The symplectic Fourier transform of the (cross) ambiguity function $\sFourier\Amb_{g\gamma}^{(\alpha)}$ 
is called the (cross) Wigner distribution of $g$ and $\gamma$ in polarization $\alpha$.

\subsection{Weyl Correspondence and Spreading Representation}

The operational meaning of pseudo-differential operators can 
be stated  with a (distributional) kernel, coordinate-based in the form of 
infinite matrices or in some algebraic manner (see for example
\cite[Chapter 14]{grochenig:gaborbook}). The kernel based 
description is usually written in a form like:
\begin{equation}
   (\BH \gamma)(t)=\int_{\Reals^n} h(t,t') \gamma(t')dt'
\end{equation}
with a kernel $h:\Reals^{2n}\rightarrow\Complexes$ (for two Schwartz functions
$\gamma,g\in\Schwarz(\Reals^n)$ the kernel $h$ exists even as a tempered distribution, i.e.
Schwartz kernel theorem states $h\in\Schwarz'(\Reals^{2n})$ with
$\langle g,\BH\gamma\rangle=\langle h,\bar{g}\otimes\gamma\rangle$,
see for example \cite[Thm. 14.3.4]{grochenig:gaborbook}).
However, the abstract description of $\BH$ as superpositions of time--frequency shifts is 
important and quite close to the physical modeling of time--varying channels.
We will adopt this time-frequency framework to describe the channel operators.
Let us denote with  $\schattenclass_\infty$ the set of compact operators, i.e. for 
$X\in\schattenclass_\infty$ holds $X=\sum_k s_k\langle x_k,\cdot\rangle y_k$ with
singular values $\{s_k\}$ and two orthonormal bases (singular functions) $\{x_k\}$ and $\{y_k\}$. 
For $p<\infty$ the  $p$th Schatten class is the set of operators
$\schattenclass_p:=\{X\,\,|\,\,\lVert X\rVert^p_p:=\Trace{((X^*X)^{p/2})}=\sum_k |s_k|^p<\infty\}$
where $\Trace(\cdot)$ is the usual meaning of the trace (e.g. evaluated 
in a particular basis).
Then $\schattenclass_p$ for $1\leq p<\infty$ are Banach spaces and 
$\schattenclass_1\subset\schattenclass_p\subset\schattenclass_\infty$ 
(see for example \cite{reed:simon:fourier}). 
The sets $\settraceclass$ and $\schattenclass_2$ 
are called trace class and Hilbert--Schmidt operators. Hilbert--Schmidt operators
form itself a Hilbert space with inner product
$\langle Y,X\rangle_{\schattenclass_2}:=\Trace(Y^*X)$.
For $X\in\settraceclass$ it holds by
properties of the trace that
$|\langle Y,X\rangle_{\schattenclass_2}|\leq\lVert X\rVert_1\lVert Y\rVert$, where
$\lVert\cdot\rVert$ denotes the operator norm.
Hence for $Y=\Shift_\mu^{(\alpha)}$ given by \eqref{eq:weyl:shift:alphageneralized} one can 
define analogously to the ordinary Fourier transform \cite{daubechies:integraltransform,grossmann:wignerweyliso}
a mapping $\settraceclass\rightarrow\Leb{2}(\Reals^{2n})$ via:
\begin{equation}
   \BHSpread_{X}^{(\alpha)}(\mu)\defeq
   \langle\Shift_\mu^{(\alpha)},X\rangle_{\schattenclass_2}   
   \label{eq:tfanalysis:noncomm:fourier}
\end{equation}
In essence, the kernel $h$ of the channel operator $\BH$ is given as
the (inverse) Fourier transform in the $\mu_2$ variable (see for example
\cite[Chapter 14]{grochenig:gaborbook}).
Note that $\BHSpread_{X}^{(\alpha)}(0)=\Trace(X)$ and 
$|\BHSpread_{X}^{(\alpha)}(\mu)|\leq\lVert X\rVert_1$ 
(because $\lVert\Shift_\mu^{(\alpha)}\rVert=1$). The function
$\BHSpread_{X}^{(\alpha)}$ is sometimes
called the ''non--commutative'' Fourier transform \cite{holevo:propquantum}, 
characteristic function, inverse Weyl transform \cite{weyl:theoryofgroups:quantum} 
or $\alpha$--generalized \emph{spreading function} of $X$ \cite{kozek:thesis,kozek:generalweyl}.
From \eqref{eq:weyl:shift:alphageneralized} it follows that
$\BHSpread_{X}^{(\alpha)}=e^{i2\pi(1/2-\alpha)\asy(\mu,\mu)}\BHSpread_{X}^{(1/2)}$.
\begin{mylemma}[Spreading Representation]
   Let $X\in\schattenclass_2$. Then there it holds:
   \begin{equation}
      X
      =\int_{\Reals^{2n}} \langle\Shift_\mu^{(\alpha)},X\rangle_{\schattenclass_2}\Shift_{\mu}^{(\alpha)} d\mu
      \label{eq:tfanalysis:linop:spreading}
   \end{equation}
   where the integral is meant in the weak sense\footnote{
     For $\lVert\BHSpread_{X}^{(\alpha)}\rVert_1<\infty$
     \eqref{eq:tfanalysis:linop:spreading} is a Bochner integral.
     Weak interpretation of \eqref{eq:tfanalysis:linop:spreading} as $\langle g,X\gamma\rangle$ 
     extents the meaning of this integral
     to tempered distributions
     \cite[Chapter 2]{folland:harmonics:phasespace} or 
     \cite[Chapter 14.3]{grochenig:gaborbook}.}. 
   \label{lemma:tfanalysis:linop:spreading}
\end{mylemma}
The extension  to the Hilbert--Schmidt operators $\schattenclass_2$ is
due to continuity of the mapping in \eqref{eq:tfanalysis:noncomm:fourier} and density of 
$\schattenclass_1$ in $\schattenclass_2$.
A complete proof of this lemma can be found in many books on Weyl calculus 
(for example matched to our notation in \cite[Chapter V]{holevo:propquantum} and \cite{Segal1963}).
Furthermore, the following 
important shift--property: 
\begin{equation}
   \BHSpread^{(\alpha)}_{\Shift_\mu(\beta)X\Shift_\mu^*(\beta)}(\nu)=
   e^{-i2\pi\sympl(\mu,\nu)}\BHSpread^{(\alpha)}_{X}(\nu)
   \label{eq:tfanalysis:noncomm:shiftproperty}
\end{equation}
can be verified easily using \eqref{eq:weyl:shift:alphageneralized} and \eqref{eq:weyl:shift:weylcommrelation}.
The composition of the symplectic Fourier transform
$\sFourier$ as defined in \eqref{eq:tfanalysis:symplectic:fourier} 
with the mapping in \eqref{eq:tfanalysis:noncomm:fourier}
establishes the so called \emph{Weyl correspondence} \cite{weyl:theoryofgroups:quantum}
in a particular polarization $\alpha$ (for this generalized approach in signal
processing see also \cite{kozek:generalweyl}). 
The function $\BHWeyl^{(\alpha)}_{X}=\sFourier\BHSpread^{(\alpha)}_X$ is called (generalized) \emph{Weyl symbol} of $X$.
The original Weyl symbol is $\BHWeyl^{(0)}_{X}$. The cases $\alpha=\tfrac{1}{2}$ and
$\alpha=-\tfrac{1}{2}$ are also known as Kohn--Nirenberg symbol (or Zadeh's 
time--varying transfer function) and Bello's 
frequency--dependent modulation function \cite{bello:wssus}. 
The Parseval identities are:
\begin{equation}
   \langle X,Y\rangle_{\schattenclass_2}=
   \langle \BHSpread^{(\alpha)}_X,\BHSpread^{(\alpha)}_Y\rangle=
   \langle \BHWeyl_X^{(\alpha)},\BHWeyl_Y^{(\alpha)}\rangle
   \label{eq:tfanalysis:weyl:parceval}
\end{equation}
for $X,Y\in\schattenclass_2$.
For a rank--one operator $X=\langle\gamma,\cdot\rangle g$ it follows that
$\BHSpread^{(\alpha)}_X=\bar{\Amb}_{g\gamma}^{(\alpha)}$
such that \eqref{eq:tfanalysis:weyl:parceval} reads in this case as:
$\langle g,Y\gamma\rangle=\langle\bar{\Amb}_{g\gamma}^{(\alpha)},\BHSpread^{(\alpha)}_Y\rangle$.

%
\section{Problem Statement and Main Results}
\label{sec:mainresults}
%
In this section we will establish a concept, which we have called 
the ''approximate eigenstructure''. The latter are sets of signals
and coefficients which fulfill a particular property
of singular values and functions up to certain approximation error $E_p$
measured in $p$--norm.
Part \ref{subsec:approxeigen} motivates this concept for a single channel operator.
In part \ref{subsec:approxeigen:compacteigen} of this section we will then extent
this framework to a
time--frequency formulation for ''random''
time--varying channels with a common support of the spreading functions. 
We consider on how approximate eigenstructure behavior scales with the respect
to the particular spreading functions, which is the main problem of this paper.
Recent results in this direction are for $p=2$ and based on estimates on the approximate
product rule of Weyl symbols. We will give a general formulation of this approach 
and an overview over the known results for $E_2$ in part \ref{subsec:approxeigen:symbolcalculus} of this section. 
After that we present in \ref{subsec:approxeigen:directapproach}
a new (direct) approach for upperbounding $E_p$ yielding for our setup also
improved and more general estimates for $p=2$.
This part contains a summary of the main results, where the more 
detailed analysis is in Section \ref{sec:genanalysis}.

\subsection{The Approximate Eigenstructure}
\label{subsec:approxeigen}
It is a common approach to describe a given channel operator $\BH$ 
on a superposition of signals where $\BH$ act rather simple. 
As already mentioned, compact operators on a Hilbert space can be formally represented as
$\BH=\sum_{k=1}^\infty s_k\langle x_k,\cdot\rangle y_k$ 
with the singular values $\{s_k\}$ and singular functions
$\{x_k\}$ and $\{y_k\}$. 
Transmitting an information bearing complex data symbol $c$ for example in the form
of the signal $s=c\cdot x_k$ through $\BH$ we known that with proper channel measurement 
(obtaining $s_k$) the information can be coherently ''recovered'' from the estimate
$\langle y_k,\BH s\rangle=s_k\cdot c$. 
The crucial point here is that the transmitting device has to know and implement $\{x_k\}$ before.
However, in practical implementation $\{x_k\}$ is required to be fixed and structured to some sense
(for example in the form  of filterbanks).
But in general, also the singular functions depend explicitly on the operator 
$\BH$, i.e. they vary from one realization to another. They can be very unstructured and it is 
difficult to relate properties of $\BH$ in such representations to physical measurable quantities.

Hence, instead of requiring $\BH x_k=s_ky_k$ we would like to have that
$\BH x_k-s_ky_k$ is ''small'' in some sense. Usually, approximation in the $\Leb{2}$--norm 
seems to be of most interest in the signal design. However, there are 
certain problems as peak power and stability issues where stronger results are required.
Furthermore intuitively we are aware that the approximation of the singular behavior
of $\{s_k,y_k,x_k\}$ has to be ''uniform'' in more than one particular norm.
In this paper we consider the $\Leb{p}$ norms for the approximation, thus we have
the following formulation for the Hilbert space $\Leb{2}(\Reals^n)$:
\begin{mydefinition}[Approximate Eigenstructure]
   Let $\epsilon$ be a given positive number. Consider $\lambda\in\Complexes$ and
   two functions $g,\gamma\in\Leb{2}(\Reals^n)$ 
   with $\lVert g\rVert_2=\lVert \gamma\rVert_2=1$. If
   \begin{equation}
      E_p:=\lVert \BH\gamma-\lambda g\rVert_p\leq\epsilon
      \label{eq:approxeigen1}
   \end{equation}
   we call $\{\lambda,g,\gamma\}$ a $\Leb{p}$--approximate eigenstructure of $\BH$
   with bound $\epsilon$.
   \label{def:approxeigen:Ep}
\end{mydefinition}
The set of $\lambda$'s for which exists $g_\lambda$ such that 
$\{\lambda,g_\lambda,g_\lambda\}$ is a $\Leb{2}$--approximate eigenstructure for a common fixed 
$\epsilon$ is also called the \emph{$\epsilon$--pseudospectrum}\footnote{Thanks to 
  T. Strohmer for informing me about this relation.} of $\BH$.
More generally, we will allow also  $g\neq \gamma$ such that
the term ''approximate singular'' functions is suited for our approach as well.
Obviously for the ''true eigenstructure'' $\{s_k,y_k,x_k\}$ as defined above we have that $\epsilon=0$ for each $p$ and $k$
. 
On the other hand, for given $g$ and $\gamma$ the minimum of the left hand side (lhs) of \eqref{eq:approxeigen1} 
is achieved for $p=2$ at $\lambda=\langle g,\BH\gamma\rangle$ such that $E_p$ for
$\{\langle g,\BH\gamma\rangle,g,\gamma\}$ describes the amount of 
\emph{orthogonal distortion} caused  by $\BH$ measured in the $p$--norm.

\subsection{The Problem Statement for Channels with Compactly Supported Spreading }
\label{subsec:approxeigen:compacteigen}
It is of general importance to what degree the Weyl symbol or a smoothed version 
of it approaches the 
eigen--value (or more generally singular value) characteristics of a given 
channel operator $\BH$. 
Inspired from the ideas in \cite{kozek:eigenstructure} we will consider now 
the following question: What is the error $E_p(\mu)$ if we approximate the action of $\BH$ on 
$\Shift_\mu\gamma$ as a multiplication of $\Shift_\mu g$ by $\lambda(\mu)$?
Hence, instead of the ''true'' eigenstructure consisting of the singular values and functions of $\BH$ 
we shall consider a more structured family $\{\lambda(\mu),\Shift_\mu g,\Shift_\mu\gamma\}$.
The latter will intuitively probe the operator $\BH$ locally in a phase space (time--frequency) meaning
if $g$ and $\gamma$ are in some sense time--frequency localized around the origin.
The validity of this approximate picture, in which the function $\lambda:\Reals^{2n}\rightarrow\Complexes$ 
now serves as a multiplicative channel is essentially described by $E_p(\mu)$.

For example, in wireless communication $\Shift_\mu g$ and $\Shift_\mu\gamma$ could be
well--localized prototype filters at time--frequency slot $\mu$ of the receive and 
transmit filterbanks of a particular communication device
and $\lambda(\mu)$ is an effective channel coefficient to be equalized. 
However, with this application in mind, we are typically
confronted with \emph{random channel operators} $\BH$ characterized by
random spreading functions $\BHSpread_{\BH}^{(\alpha)}$ having
a common (Lebesgue measurable) support $U$ of non--zero and finite measure
$|U|$, i.e. $0<|U|<\infty$. The assumption of a known support seems to be
the minimal apriori channel knowledge that enters practically 
the system design (e.g. of a communication device).
For example, a typical doubly--dispersive 
channel model ($n=1$) for this application is that spreading
occurs in $U=[0,\taumax]\times[-B_D,B_D]$ where $\taumax$ and $B_D$ are the maximum 
delay spread and Doppler frequency. It is then desirable to have \emph{common
prototype filters for all these channel realizations}.
It is clear that in this direction
Definition \ref{def:approxeigen:Ep} is not yet adequate enough. We have to measure 
the approximation error with respect to a certain scale of the 
particular random spreading functions. In this paper we measure the approximate 
eigenstructure with respect to its $\Leb{q}$--norm. 
We believe that this approach is important  to have reasonable
estimates for the various statistical fading and scattering environments. An example
of such an application is given in Remark \ref{rem:approxeigen:statisticalmodel}
in Section \ref{sec:numerical}.

We consider only bounded 
spreading functions such that
the operator $\BH$ is of Hilbert Schmidt type, i.e. $\BH\in\schattenclass_2$.
To this end, let us call this set of channel operators as $\OpClass(U)$, i.e.
\begin{equation}
   \OpClass(U):=\{\,\BH\,| \support{\BHSpread_{\BH}^{(\alpha)}}\subseteq U\,\text{and}\,
   \sup_{\mu\in U} |\BHSpread_{\BH}^{(\alpha)}(\mu)|<\infty\}
\end{equation}
As already discussed for example in \cite{Pfander08:opid} the operator
class $\OpClass(U)$ does not include limiting cases of doubly--dispersive channels like 
the time--invariant channel or the identity. Generalizations, 
for example in the sense of tempered distributions, are beyond the scope
of this paper. We aim at an extension of Definition \ref{def:approxeigen:Ep} for the approximate
eigenstructure which is meaningful and suited for this class of channels. 
We will formulate this as our main problem of this paper:
\begin{myproblem}
   Consider two functions $g,\gamma:\Reals^n\rightarrow\Complexes$ 
   with $\lVert g\rVert_2=\lVert\gamma\rVert_2=1$. Let be $1\leq q\leq\infty$,
   $1\leq p<\infty$ and $0<\delta<\infty$ such that 
   for all operators $\BH\in\OpClass(U)$ it holds:
   \begin{equation}
      \begin{split}
         E_p(\mu):
         &=\lVert\BH\Shift_\mu\gamma-\lambda(\mu)\Shift_\mu g\rVert_p
         \leq \delta\cdot \lVert\BHSpread_{\BH}^{(\alpha)}\rVert_q
      \end{split}
      \label{eq:approxeigen:epweyl}
   \end{equation}
   where the $p$--norm is with respect to the argument of the function
   $\BH\Shift_\mu\gamma-\lambda(\mu)\Shift_\mu g$.
   Then $\{\lambda(\mu),\Shift_\mu g,\Shift_\mu\gamma\}$ is an $\Leb{p}$-approximate eigenstructure for
   {\bf all} $\BH\in\OpClass(U)$, each of them with 
   individual bound $\epsilon=\delta\cdot \lVert\BHSpread_{\BH}^{(\alpha)}\rVert_a$.
   How small can we choose the scale $\delta$ given $g$, $\gamma$, $U$, $p$ and $q$? 
   What can be said about $\inf_{g,\gamma}(\delta)$?
   \label{problem:approxeigen:epweyl}
\end{myproblem}
Note that, independently of the polarization $\alpha$, the operator
$\Shift_\mu$ can be replaced in \eqref{eq:approxeigen:epweyl}
with any $\beta$--polarized shift
$\Shift_\mu^{(\beta)}$ without change of $E_p(\mu)$.
Furthermore, as already stated in the definition of $E_p$ in \eqref{eq:approxeigen1} 
$\lVert g\rVert_2=\lVert\gamma\rVert_2=1$, throughout the rest of the paper.
Summarizing: How much could $\{\Shift_\mu g,\Shift_\mu\gamma\}$ serve as common
approximations (measured in $p$--norm) to the singular functions of the operator class $\OpClass(U)$
for fixed $U$ ? 

\subsection{Results Based on the Approximate Product Rule}
\label{subsec:approxeigen:symbolcalculus}
In previous work \cite{kozek:thesis,kozek:transferfunction,matz:thesis,matz:timefreq:transfer} results 
were provided for $g=\gamma$ and (apart of \cite{matz:timefreq:characterization})
$\lambda=\BHWeyl^{(\alpha)}_{\BH}$ for the case $p=2$.
These are obtained if one considers the problem from view of symbolic calculus
and can be summarized in the following lemma:
\begin{mylemma}
   Let $\gamma=g$ and $\lambda=\BHWeyl^{(\alpha)}_{\BH}$. It holds:
   \begin{equation}
      E_2(\mu)
      \leq\left(|\BHWeyl^{(\alpha)}_{\BH^*\BH}(\mu)-|\BHWeyl^{(\alpha)}_{\BH}(\mu)|^2|
        +\lVert\BHSpread^{(\alpha)}_{\BH^*\BH}\Omega\rVert_1+2|\BHWeyl^{(\alpha)}_{\BH}(\mu)|\cdot\lVert\BHSpread^{(\alpha)}_{\BH} \Omega\,\rVert_1\right)^{\frac{1}{2}}
        \label{eq:gmc:underspread:generalupperbound}
   \end{equation}
   where 
   $\Omega=|\Amb^{(\alpha)}_{\gamma\gamma}-1|$.
   \label{lemma:gmc:underspread:generalupperbound}
\end{mylemma}
Note that the lemma not yet necessarily requires  $\BH\in\OpClass(U)$. 
The proof is provisioned in Appendix \ref{appendix:lemma:gmc:underspread:generalupperbound:proof}.
This bound is motivated by the work of W. Kozek \cite{kozek:thesis}. However
it has been formulated in a more general context. The first term of the bound in 
\eqref{eq:gmc:underspread:generalupperbound} contains
the Weyl symbol $\BHWeyl_{XY}^{(\alpha)}$ of the composition $XY$ of two operators $X$ and $Y$
($\BH^*$ and $\BH$ in this case), which
is the twisted multiplication \cite{pool:mathweylcorrespondence} of the symbols of the operators $X$ and $Y$.
On the level of spreading functions\footnote{%
  Symbols (spreading functions) in
  $\Leb{2}(\Reals^{2n})$ with twisted multiplication (convolution) are 
  $*$-isomorph to the algebra of Hilbert--Schmidt operators. 
},
$\BHSpread_{XY}^{(\alpha)}$ is given by the so called  
\emph{twisted convolution} $\twistconv_\phi$  of $\BHSpread_X^{(\alpha)}$ and $\BHSpread_Y^{(\alpha)}$
\cite{Segal1963,Kastler1965}:
\begin{equation}
   \begin{split}
      &\BHSpread^{(\alpha)}_{XY}(\rho)=
      \int_{\Reals^{2n}} \BHSpread^{(\alpha)}_{X}(\mu)\BHSpread^{(\alpha)}_{Y}(\rho-\mu)e^{-i2\pi\phi(\mu,\rho)}d\mu
      \defeq(\BHSpread^{(\alpha)}_{X}\twistconv_\phi \BHSpread^{(\alpha)}_{Y})(\rho)      
   \end{split}
   \label{eq:tfanalysis:spreading:twistedconv}
\end{equation}
with $\phi(\mu,\rho)=(\alpha+\frac{1}{2})\asy(\mu,\rho)+(\alpha-\frac{1}{2})\asy(\rho,\mu)-2\alpha\asy(\mu,\mu)$.
For the polarization $\alpha=0$ it follows $\phi(\mu,\rho)=\sympl(\mu,\rho)/2$ and
conventional convolution is simply $\twistconv_0$.
Expanding $\exp(-i2\pi\phi(\mu,\rho))$ in $\mu$  as a Taylor series reveals that twisted convolutions
are weighted sums of $\twistconv_0$--convolutions \cite{folland:harmonics:phasespace} 
related to moments of $\BHSpread^{(\alpha)}_{X}$ and $\BHSpread^{(\alpha)}_{Y}$.
Hausdorff--Young inequality with sharp constants 
$c^2_p=p^{\frac{1}{p}}/p'^{\frac{1}{p'}}$ (and $c_1=c_\infty=1$)  gives for $1\leq p\leq 2$ 
estimates on the following ''approximate product rule'' of Weyl symbols:
$\lVert\BHWeyl^{(\alpha)}_{XY}-\BHWeyl^{(\alpha)}_X\BHWeyl^{(\alpha)}_Y\rVert_{p'}
\leq2c^{2n}_p \lVert F\rVert_p$ where 
$F(\rho)=\int_{\Reals^{2n}} |\BHSpread^{(\alpha)}_{X}(\mu)\BHSpread^{(\alpha)}_{Y}(\mu-\rho)\sin(2\pi\phi(\mu,\rho))|d\mu$.
In particular, for  $p=1$ we get:
\begin{equation}
   \begin{split}
      |\BHWeyl^{(\alpha)}_{\BH^*\BH}-|\BHWeyl^{(\alpha)}_{\BH}|^2|\leq 
      2\lVert F\rVert_1\qquad\text{a.e.}
   \end{split}
   \label{eq:gmc:underspread:multcalc}
\end{equation}
Let us assume now that $\BH\in\OpClass(U)$.
With $\chi_U$ we shall denote the characteristic function of $U$ (its indicator function).
Kozek \cite[Thm. 5.6]{kozek:thesis}  has considered the case
$\alpha=0$ obtaining the following result:
\begin{mytheorem}[W. Kozek \cite{kozek:thesis}]
   Let $U=[-\tau_0,\tau_0]\times[-\nu_0,\nu_0]$  and $\alpha=0$. If
   $|U|=4\tau_0\nu_0\leq 1$ then
   \begin{equation}
      E_2(\mu)\leq\left(
      2\sin(\frac{\pi |U|}{4})\lVert\BHSpread^{(0)}_{\BH}\rVert^2_1+
      \epsilon_\gamma\left(\lVert\BHSpread^{(0)}_{\BH^*\BH}\rVert_1+
      2\lVert\BHSpread^{(0)}_{\BH}\rVert^2_1\right)\right)^{\frac{1}{2}}
    \label{eq:approxeigen:kozektheorem}
   \end{equation}
   where $\epsilon_\gamma=\lVert(\Amb^{(0)}_{\gamma\gamma}-1)\chi_U\rVert_\infty$.    
   \label{thm:approxeigen:kozektheorem}
\end{mytheorem}
The proof can be found in \cite{kozek:thesis} or independently from 
Lemma \ref{lemma:gmc:underspread:generalupperbound} with $\epsilon_\gamma=\lVert\Omega\chi_U\rVert_\infty$ and 
$\lVert\BHWeyl_{\BH}^{(\alpha)}\rVert_\infty\leq\lVert\BHSpread_{\BH}^{(\alpha)}\rVert_1$.
Further utilizing the fact 
that $\lVert\BHSpread^{(\alpha)}_{\BH^*\BH}\rVert_1\leq\lVert\BHSpread^{(\alpha)}_{\BH}\rVert^2_1$, 
Equation \eqref{eq:approxeigen:kozektheorem} can be written as:
\begin{equation}
   \frac{E_2(\mu)}{\lVert\BHSpread^{(0)}_{\BH}\rVert_1}
   \leq \left(2\sin(\frac{\pi |U|}{4}) + 3\lVert(\Amb^{(0)}_{\gamma\gamma}-1)\chi_U\rVert_\infty\right)^{\frac{1}{2}}
   \label{eq:approxeigen:kozektheorem:simplified}
\end{equation}
which gives an initial answer to the problem formulated in section \ref{problem:approxeigen:epweyl}.
Theorem \ref{thm:approxeigen:kozektheorem} was further extended in \cite[Thm. 2.22]{matz:thesis} 
by G. Matz (see also \cite{matz:timefreq:transfer}) to a formulation in terms of weighted $1$--moments of
(not necessarily compactly supported) spreading functions. His approach
includes also different polarizations $\alpha$. 
For a spreading function as in Theorem \ref{thm:approxeigen:kozektheorem} and $\alpha=0$ the results
agree with \eqref{eq:approxeigen:kozektheorem:simplified}.
Equation \eqref{eq:approxeigen:kozektheorem:simplified} could be interpreted 
in such a way that only the second term can be controlled
by $\gamma$ (e.g. pulse shaping) where the first term of the rhs of \eqref{eq:approxeigen:kozektheorem:simplified}
is only related to the overall spread $|U|$. However we shall show in the next section that the first
term can be eliminated 
from the bound and the second (shape--independent) term can further tightened.

\subsection{Results Based on a Direct Approach}
\label{subsec:approxeigen:directapproach}
\newcommand{\brho}{\bar{\rho}}
\newcommand{\brhoi}{{\bar{\rho}_\infty}}
We have considered the function $\lambda$ as the Weyl symbol 
exclusively for the exponents $p=2$ and $a=1$ in \ref{subsec:approxeigen:symbolcalculus}. This approach
is in line with prior work of Kozek, Matz and Hlawatsch and provides
results on the approximate eigenstructure problem established in Section \ref{problem:approxeigen:epweyl}.
To obtain further results for different values of $p$, $a$ and $\lambda$,
we shall now restart the analysis from a different perspective. In the following
we present the main results, through most of the analysis will be presented in Section 
\ref{sec:genanalysis}. We use a ''smoothed'' version
of the Weyl symbol:
\begin{equation}
   \lambda=\sFourier (\BHSpread^{(\alpha)}_{\BH}\cdot B)
   \label{eq:approxeigen:lambda}
\end{equation}
where $B:\Reals^{2n}\rightarrow\Complexes$ is a bounded function.
We consider two important cases:

{\bf Case C1:} Let $B=\Amb^{(\alpha)}_{g\gamma}$ such that \eqref{eq:approxeigen:lambda}
reads as $\lambda=\BHWeyl_{\BH}^{(\alpha)}\ast\sFourier\Amb^{(\alpha)}_{g\gamma}$ where 
$\ast$ denotes convolution. This corresponds to the well known 
smoothing with the cross Wigner function $\sFourier\Amb^{(\alpha)}_{g\gamma}$
and was already considered in \cite{matz:timefreq:characterization} (for averages
over WSSUS\footnote{Wide--sense stationary uncorrelated scattering (WSSUS) 
channel model \cite{bello:wssus}}
channels ensembles). In particular
this is exactly the orthogonal distortion: 
\begin{equation}
   \lambda(\mu)=\langle\Shift_\mu g,\BH\Shift_\mu\gamma\rangle
   =\left(\sFourier(\BHSpread^{(\alpha)}_{\BH}\cdot \Amb_{g\gamma}^{(\alpha)})\right)(\mu)
\end{equation}
as already mentioned in \ref{subsec:approxeigen}  and corresponds to the choice of the $E_2$--minimizer. 
Since $\lambda$ depends in this case  on $g$ and $\gamma$ we
consider here how accurately the action of operators $\BH$ on the family $\{\Shift_\mu\gamma\}$
can be described as multiplication operators on the family $\{\Shift_\mu g\}$. 
From the rule in 
\eqref{eq:weyl:shift:alphageneralized} the definition of the cross ambiguity function
in \eqref{eq:tfanalysis:crossamb} and the non--commutative Fourier
transform in \eqref{eq:tfanalysis:noncomm:fourier}  
it is clear that this choice is also independent of the polarization
$\alpha$. Recall that the Weyl symbol of a rank-one operator is the Wigner distribution, such
that with this approach we again effectively compare twisted with ordinary convolution.

{\bf Case C2:} Here we consider $B=1$ such that $\lambda=\BHWeyl_{\BH}^{(\alpha)}$ is the Weyl symbol. 
The function $\lambda$ is now independent of $g$ and $\gamma$.
Thus in contrast to C1 this case is related to the ''pure'' symbol calculus.
Obviously, we have to expect now a dependency on the polarization 
$\alpha$. Furthermore, this was the approach considered for $p=2$ in the previous part\footnote{However,
  also there the same methodology as in \eqref{eq:approxeigen:lambda} could be applied as well.} of this section.

The first theorem parallels Theorem \ref{thm:approxeigen:kozektheorem} 
and its consequence
\eqref{eq:approxeigen:kozektheorem:simplified}. We shall not yet restrict ourselves to the cases
C1 and C2. Instead we only require that $B$ has to be essentially bounded.
\begin{mytheorem}
   Let $\BH\in\OpClass(U)$,
   $g,\gamma\in\Leb{\infty}(\Reals^n)$ and $B\in\Leb{\infty}(\Reals^{2n})$. 
   For $2\leq p<\infty$ and $1\leq q\leq\infty$ (with the usual meaning for 
   $q=\infty$):
   \begin{equation}
      \frac{E_p(\mu)}{\lVert\BHSpread^{(\alpha)}_{\BH}\rVert_q}\leq
      C^{\frac{p-2}{p}}\cdot\lVert(1+|B|^2-2\Real{\Amb^{(\alpha)}_{g\gamma}\bar{B}})\chi_U\rVert^{1/p}_{q'/p}
      \label{eq:thm:approxeigen:Ep1}
   \end{equation}
   where $C$ is a constant depending on $g$, $\gamma$ and $B$. 
   The minimum of this bound over $B$ is achieved in the case of $p=2$ for C1.
   \label{thm:approxeigen:Ep1}
\end{mytheorem}
\begin{myproof}
   The proof follows from the middle term of \eqref{eq:lemma:approxeigen:lemmaep3:eq1} 
   in Lemma \ref{lemma:approxeigen:lemmaep3} if we set $C$ as:
   \begin{equation}
      \begin{split}
         C
         &=\underset{x\in\Reals^n,\,\, \nu\in U}{\esup}|(\Shift_\nu^{(\alpha)}\gamma)(x)-B(\nu)g(x)|
         \leq \lVert\gamma\rVert_\infty+\lVert B\rVert_\infty\lVert g\rVert_\infty
      \end{split}
   \end{equation}
   In Lemma \ref{lemma:approxeigen:lemmaep3} the range of $p$ is $1\leq p<\infty$. However, from the 
   discussion in Section \ref{subsec:uniformestimates} it is clear that 
   \eqref{eq:thm:approxeigen:Ep1} gives
   only for $p\geq 2$ a reasonable bound.
\end{myproof}

{\bf Comparison to the bound of Kozek:}
With $|1-\Real{\Amb_{\gamma\gamma}^{(\alpha)}}|\leq |1-\Amb_{\gamma\gamma}^{(\alpha)}|$
we can transform the result of the last theorem for C2 with settings $p=2$, $q=1$ 
and $g=\gamma$ into a form comparable to \eqref{eq:approxeigen:kozektheorem} 
and \eqref{eq:approxeigen:kozektheorem:simplified} which is:
  \begin{equation}
     \frac{E_2(\mu)}{\lVert\BHSpread^{(\alpha)}_{\BH}\rVert_1}\leq 
     \left(2\lVert (1-\Amb^{(\alpha)}_{g\gamma})\chi_U\rVert_\infty\right)^{\frac{1}{2}}
\end{equation}
Hence this technique improves the previous bounds.
It includes different polarizations $\alpha$ and does not require any
shape or size constraints on $U$.
Interestingly the offset in \eqref{eq:approxeigen:kozektheorem:simplified},
which does not depend on $(g,\gamma)$
and in an initial glance seems to be related to the notion of underspreadness,
has been disappeared now. 

{\bf Discussion of the critical size:}
The behavior of the bound in \eqref{eq:thm:approxeigen:Ep1}  on $|U|$ depends in general on the choice
of the function $B$. For example for the case C1, $p=q=2$  and with \eqref{eq:tfanalysis:crossamb:properties}
it follows that  the rhs of \eqref{eq:thm:approxeigen:Ep1} is the square root of
$|U|-\langle |\Amb_{g\gamma}^{(\alpha)}|^2,\chi_U\rangle$ and again with \eqref{eq:tfanalysis:crossamb:properties}
we have that:
\begin{equation}
   \sqrt{|U|-\min(|U|,1)}\leq\text{rhs of \eqref{eq:thm:approxeigen:Ep1}}\leq \sqrt{|U|}.
\end{equation}
This implies that this term is of the same order as $\sqrt{|U|}$ for $|U|\gg1$ 
(see also Lemma \ref{lemma:approxeigen:ep:uniform} later on). The lhs of the inequality suggests that 
for $|U|\leq 1$ the scaling behavior might alter, i.e. $|U|=1$ is in this sense
a critical point between over- and underspread channels as introduced 
in \cite{kozek:thesis}.
On the other hand the lhs of the last equation is not zero for $0<|U|\leq1$. 
Indeed from Theorem \ref{corr:ambbound:supp} 
we have an improved version as follows:
\begin{equation}
   \sqrt{|U|-\min(|U|e^{-\frac{|U|}{e}},1)}\leq\text{rhs of \eqref{eq:thm:approxeigen:Ep1}}\leq \sqrt{|U|}
\end{equation}
which suggest that at $|U|=e$ the behavior changes.

{\bf Restriction to the cases C1 and C2:} If we further restrict ourselves to $q>1$ (i.e. $q'<\infty$) we can
establish the relation to weighted norms of ambiguity functions \cite{jung:isit06}. For
simplicity let us consider now the two cases C1 ($k=1$) and C2 ($k=2$). We define therefore
the functions $A_k:\Reals^{2n}\rightarrow\Reals$ for $k=1,2$ as:
\begin{equation}
   A_1:=|\Amb_{g\gamma}^{(\alpha)}|^2\quad\text{and}\quad A_2:=\Real{\Amb_{g\gamma}^{(\alpha)}}
   \label{eq:approxeigen:Ak}
\end{equation}
We then have the following result:
\begin{mytheorem}
   Let $\BH\in\OpClass(U)$ and
   $g,\gamma\in\Leb{\infty}(\Reals^n)$. 
   For  $2\leq p<\infty$, $1<q\leq\infty$ and $|U|\leq 1$ it holds:
   \begin{equation}
      \frac{E_p(\mu)}{\lVert\BHSpread^{(\alpha)}_{\BH}\rVert_{q}}\leq 
      C^\frac{p-2}{p}
      k\left(k(|U|-\langle A_k,\chi_U\rangle)\right)^{1/\max(q',p)}
   \end{equation}
   where $k=1$ for C1 and $k=2$ for C2.
   \label{thm:approxeigen:Ep2}
\end{mytheorem}
\begin{proof}
   We now use the bound \eqref{eq:lemma:approxeigen:lemmaep3:eq2} 
   in Lemma \ref{lemma:approxeigen:lemmaep3} with the uniform estimates 
   $C_{bp}\leq k$ from Lemma \ref{lemma:approxeigen:cpb:uniform}. Again, as follows from the discussion
   in Section \ref{subsec:uniformestimates}, we consider only $p\geq 2$.
\end{proof}
The assumption $|U|\leq1$ is only used to simplify the bound. Improved estimates follow
from Lemma \ref{lemma:approxeigen:lemmaep3} directly.
From the positivity of $A_1$ we observe that the orthogonal distortion
(the case C1) is always related to weighted $2$--norms of the cross ambiguity
function (the weight is in this case only $\chi_U$). For the
case C2 this can be turned into weighted $1$--norm if 
$A_2$ is positive on $U$ or fulfill certain cancellation properties. Furthermore,
the case C2 depends obviously on the polarization $\alpha$.
For particular symmetries of $U$ explicit values can be found 
as shown with the following theorem (for simplicity we consider $n=1$):
\begin{mytheorem}
   If $\BH\in\OpClass(U)$ with $\chi_U(\mu)=\chi_U(-\mu)$ and 
   $g,\gamma\in\Leb{\infty}(\Reals)$ and $B\in\Leb{\infty}(\Reals^{2})$. 
   For  $2\leq p<\infty$, $1<q\leq\infty$ and $|U|\leq 1$ it holds:
   \begin{equation}
      \frac{E_p(\mu)}{\lVert\BHSpread^{(\alpha)}_{\BH}\rVert_{q}}\leq 
      32^\frac{p-2}{4p}
      k\left(k|U|(1-L)\right)^{1/\max(q',p)}
      \label{eq:thm:approxeigen:Ep3}
   \end{equation}
   In general, $L\geq\lambda_{\max}(Q^*Q)^{1/k}$ and $Q$ 
   is an operator with spreading function $\chi_U/|U|$ in
   polarization $0$. Furthermore, 
   $L\geq l(|U|2/k)$ with $l(x)=2(1-e^{-x/2})/x$ for $U$ being a disc and
   $l(x)=2\cdot\text{erf}(\sqrt{\pi x/8})^2/x$ for $U$ being a square. 
   \label{thm:approxeigen:Ep3}
\end{mytheorem}
\begin{proof}
   We combine Lemma \ref{lemma:approxeigen:lemmaep3} and 
   Lemma \ref{lemma:approxeigen:localization} with the uniform estimates 
   $C_{bp}\leq k$ from Lemma \ref{lemma:approxeigen:cpb:uniform} such that
   \begin{equation}
      \frac{E_p(\mu)}{\lVert\BHSpread^{(\alpha)}_{\BH}\rVert_{q}}\leq 
      32^\frac{p-2}{4p}
      k\left(k|U|(1-\lambda_{\max}(Q^*Q)^{1/k})\right)^{1/\max(q',p)}
   \end{equation}
   where $Q$ is a compact operator with spreading function $\chi_U/|U|$ in polarization $\alpha$.
   From Lemma \ref{lemma:approxeigen:infR1:symmetric} we know 
   that our assumptions imply that $Q$ is Hermitian for $\alpha=0$. In general, therefore it holds 
   that $L=\lambda_{\max}(Q^*Q)^{1/k}=\lambda_{\max}(Q)^{2/k}$ where $\lambda_{\max}(Q)$
   is at least as the value of the integral \eqref{eq:corr:approxeigen:lambdan} over the first Laguerre function.
   We abbreviate $L=l(|U|)^{2/k}$ such that
   for a disc of radius $\sqrt{|U|/\pi}$ this integral is
   $l(x)=2(1-e^{-x/2})/x$ and for
   a square of length $\sqrt{U}$ we get $l(x)=2\cdot\text{erf}(\sqrt{\pi x/8})^2/x$.   
   However, Lemma \ref{lemma:approxeigen:infR1:symmetric} asserts this as an upper
   bound achieved in this case 
   with Gaussians $g(x)=2^{1/4}e^{-\pi\langle x,x\rangle}$ which is tight for C2  but not for C1.
   This means, for C1 this can be further improved by direct evaluation on Gaussians.
   Indeed, from the proof of Corollary \ref{corr:approxeigen:gaussianbounds} we 
   know that $l(|U|\cdot2/k)\geq l(|U|)^{2/k}$ is achievable.
   For $\lVert V\rVert_\infty$ we get $\lVert V\rVert_\infty=2\lVert g\rVert_\infty=32^{1/4}$.
\end{proof}
This result can be extended in part to regions $U$ which are canonical equivalent
to discs and squares centered at the origin (see the discussion at the beginning
of Section \ref{subsec:approxeigen:laguerre}).
However, this holds in principle only for $p=2$ because such canonical transformations 
will change the constants in \eqref{eq:thm:approxeigen:Ep3}.

\section{General Analysis and Proofs}
\label{sec:genanalysis}
With the following lemma we separate
the support of the spreading function $\BHSpread^{(\alpha)}_{\BH}$ from the quantity $E_p(\mu)$. 
We shall make use of the 
non--negative function $V:\Reals^{n}\times\Reals^{2n}\rightarrow\RealsPlus$, defined as:
\begin{equation}
   V(x,\nu):=|(\Shift_\nu^{(\alpha)}\gamma)(x)-B(\nu)g(x)|\cdot \chi_U(\nu)\geq 0
   \label{eq:approxeigen:defV}
\end{equation}
and of the functionals $V_p(\nu):=\lVert V(\cdot,\nu)\rVert_p$, i.e.
the usual $p$--norms in the first argument. 
For simplifying our analysis we shall restrict ourselves 
to indicator weights $\chi_U$ (the characteristic function of $U$). 
However, the same can be repeated with slight abuse of notation using
more general weights.
\begin{mylemma}
   Let $\BH\in\OpClass(U)$,
   $1\leq p<\infty$, $1\leq q\leq\infty$. 
   If $V(\cdot,\nu)\in\Leb{p}(\Reals^n)$ for all $\nu\in U$ then it
   holds: 
   \begin{equation}
      E_p(\mu)/\lVert\BHSpread^{(\alpha)}_{\BH}\rVert_q\leq\lVert V_p\rVert_{q'}
      \label{eq:approxeigen:lemmaep1:1}
   \end{equation}   
   whenever $\BHSpread^{(\alpha)}_{\BH}\in\Leb{q}(\Reals^{2n})$ 
   and $V_p\in\Leb{q'}(\Reals^{2n})$.
   \label{lemma:approxeigen:lemmaep1}
\end{mylemma}
\begin{myproof}
   Firstly, using Weyl's commutation rule \eqref{eq:weyl:shift:weylcommrelation} and 
   the definition of $\lambda$ in 
   \eqref{eq:approxeigen:lambda} gives:
   \begin{equation}
      \begin{split}
         E_p(\mu)
         &\overset{\eqref{eq:approxeigen:epweyl}}{=}\lVert\int_{\Reals^{2n}} d\nu\,
         \BHSpread_{\BH}^{(\alpha)}(\nu)\Shift_\nu^{(\alpha)}\Shift_\mu\gamma-
         \lambda(\mu)\Shift_\mu g
         \rVert_p\\
         &\overset{\eqref{eq:weyl:shift:weylcommrelation}}{=}
         \lVert\Shift_\mu\left( 
         \int_{\Reals^{2n}} d\nu\,\BHSpread_{\BH}^{(\alpha)}(\nu)
         e^{-i2\pi\sympl(\nu,\mu)}
         \Shift_\nu^{(\alpha)}\gamma-
         \lambda(\mu) g\right)         
         \rVert_p \\         
         &\overset{\eqref{eq:approxeigen:lambda}}{=}
         \lVert\int_{\Reals^{2n}} d\nu\, \BHSpread_{\BH}^{(\alpha)}(\nu)
         e^{-i2\pi\sympl(\nu,\mu)}
         (\Shift_\nu^{(\alpha)}\gamma-B(\nu)g)
         \rVert_p\,.\\
      \end{split}
      \label{eq:approxeigen:lemmaep1:proof1}
   \end{equation}
   Note that the $p$--norm is with respect to the argument of the functions $g$ 
   and $\Shift_\nu^{(\alpha)}\gamma$.
   The last step follows because $\Shift_\mu^{(\alpha)}$ acts 
   isometrically on all $\Leb{p}(\Reals^n)$.
   Let $f:\Reals^n\times\Reals^{2n}\rightarrow\Complexes$ be the function defined as:
   \begin{equation}
      f(x,\nu):=e^{-i2\pi\sympl(\nu,\mu)}
      \BHSpread^{(\alpha)}_{\BH}(\nu)[(\Shift_\nu^{(\alpha)}\gamma)(x)-B(\nu)g(x)]
   \end{equation}   
   From $\BH\in\OpClass(U)$ (bounded spreading functions) and $V(\cdot,\nu)\in\Leb{p}(\Reals^n)$ for all $\nu\in U$ 
   it follows  
   that $f(\cdot,\nu)\in\Leb{p}(\Reals^n)$.
   Then \eqref{eq:approxeigen:lemmaep1:proof1} reads for $1\leq p<\infty$ by
   Minkowski (triangle) inequality 
   \begin{equation}
      \begin{split}
         E_p(\mu)
         &=\lVert\int_{\Reals^{2n}} d\nu f(\cdot,\nu)\rVert_p
         \leq\int_{\Reals^{2n}} d\nu\lVert f(\cdot,\nu)\rVert_p
         =\lVert \BHSpread^{(\alpha)}_{\BH}\cdot V_p\rVert_1
         \leq \lVert \BHSpread^{(\alpha)}_{\BH}\rVert_q\lVert V_p\rVert_{q'}
      \end{split}
   \end{equation}
   In the last step we used H\"older's inequality, such that 
   the claim of this lemma follows.
\end{myproof} 

\subsection{The Relation to Ambiguity Functions}
In the next lemma we shall show that 
$\lVert V_p\rVert_{q'}$ can be related to ambiguity functions, which occur 
for $p=2$. We introduce $R:\Reals^{2n}\rightarrow\RealsPlus$ as
the non--negative function:
\begin{equation}
   R:=V_2^2=(1+|\Amb^{(\alpha)}_{g\gamma}-B|^2-|\Amb^{(\alpha)}_{g\gamma}|^2)\cdot\chi_U\geq 0
   \label{eq:approxeigen:defR}
\end{equation} 
and abbreviate $R_s:=\lVert R\rVert_s$. Using \eqref{eq:approxeigen:Ak} we 
can write \eqref{eq:approxeigen:defR} for the cases C1 and C2 as $R=k(1-A_k)\cdot\chi_U$ for $k=1,2$. 
From the non--negativity of $R$ follows
that: 
\begin{equation}
   R_1=|U|+\lVert(\Amb^{(\alpha)}_{g\gamma}-B)\chi_U\rVert_2^2-\lVert\Amb^{(\alpha)}_{g\gamma}\chi_U\rVert_2^2
   \label{eq:approxeigen:R1}
\end{equation}
Hence, $R_1$ reflects an interplay between two localization criteria
in the phase space. 
In particular, we get for C1 and for C2:
\begin{equation}
   R_1=k(|U|-\langle A_k,\chi_U\rangle)
   \label{eq:approxeigen:R1:special}
\end{equation}
With the following lemma we shall explicitly provide the
relation between the bound $\lVert V_p\rVert_{q'}$ in Lemma \ref{lemma:approxeigen:lemmaep1} to the quantity $R_1$.
\begin{mylemma}
   For $1\leq p<\infty$ and $1\leq q'\leq\infty$  it
   holds (with the usual meaning for $q'=\infty$): 
   \begin{equation}
      \lVert V_p\rVert_{q'}\leq \lVert V\rVert_\infty^\frac{p-2}{p}
      \cdot R_{q'/p}^{1/p}
      \label{eq:lemma:approxeigen:lemmaep3:eq1}
   \end{equation}
   Equality is achieved for $p=2$ and then
   the minimum over $B$ of the rhs is achieved for C1.    
   For $q'<\infty$ let 
   $C_{pq}=R_\infty^\frac{q'-p}{q'p}$ for $p\leq q'$ and
   $C_{pq}=|U|^\frac{p-q'}{q'p}$ else. 
   Then it holds further that:
   \begin{equation}
      \lVert V_p\rVert_{q'}\leq
      \lVert V\rVert_\infty^\frac{p-2}{p}\cdot C_{pq}\cdot
      R_1^{1/\max(p,q')} 
      \label{eq:lemma:approxeigen:lemmaep3:eq2}
   \end{equation}
   with equality for $q'=p=2$. 
   \label{lemma:approxeigen:lemmaep3}
\end{mylemma}
The proof can be found in the Appendix \ref{appendix:lemma:approxeigen:lemmaep3:proof}.
The main reason for this lemma, in particular for the second part, is that it opens up for
case C1 the relation to weighted norms of ambiguity function (i.e. localization of $A_k$ on $U$).
However, for C2 we are also concerned with the question of positivity (and cancellation properties) in $U$.
We shall study these relations in more detail in Section \ref{subsec:localization}.

\subsection{Uniform Estimates}
\label{subsec:uniformestimates}
As already mentioned before, for ''true'' eigenstructure we have $E_p=0$ for all $p$,
such that the notion of approximate eigenstructure should be in some sense uniform in $q$ and $p$.
In the first step it is therefore necessary to validate  
uniform bounds for $C_{pq}$.
We observe that $\lVert V\rVert_\infty^\frac{p-2}{p}$ will then restrict 
the application of Lemma \ref{lemma:approxeigen:lemmaep3} only to $p\geq 2$ because 
$\lVert V\rVert_\infty$ will be in general small. For example for C2 and
$g=\gamma$ let $|U|\rightarrow 0$ 
in \eqref{eq:approxeigen:defV}. This behavior has to be expected because
the ambiguity function is a $\Leb{2}$--related construction and from 
$\Leb{2}$ boundedness one can only with
further decay conditions infer $\Leb{p}$--boundedness for $p<2$. Consequentially we
shall restrict the following analysis to
$2\leq p<\infty$ such that $\sup\lVert V\rVert_\infty^\frac{p-2}{p}=\max(\lVert V\rVert_\infty,1)$.
For $\lVert V\rVert_\infty$ we can use for example a worst case estimate of the form
$\lVert V\rVert_\infty\leq\lVert\gamma\rVert_\infty+\lVert B\rVert_\infty\cdot\lVert g\rVert_\infty
\leq\lVert\gamma\rVert_\infty+\lVert g\rVert_\infty$
which is valid for C1 
($\lVert B\rVert_\infty=\lVert\Amb_{g\gamma}\rVert_\infty\leq1$
by \eqref{eq:tfanalysis:crossamb:properties}) 
and C2.
\begin{mylemma}[Uniform Bounds for $C_{pq}$]
   For $2\leq p<\infty$ and
   $1<q\leq\infty$ it holds the uniform estimate 
   $C_{pq}\leq k$ if $q'\geq p$ 
   where $k=1$ for C1 and $k=2$ for C2. If $q'<p$ then it holds
   $C_{pq}\leq\max(|U|,1)$.
   \label{lemma:approxeigen:cpb:uniform}
\end{mylemma}
\begin{proof}
   It is easily verified that
   $\sup |U|^\frac{p-q'}{q'p}=\max(|U|,1)$ where the supremum is over all
   $1\leq q'<p$ and $2\leq p<\infty$.
   The same
   can be found also for $1\leq p<\infty$.    
   Similarly we get for the quantity $R_\infty^\frac{q'-p}{q'p}$ the
   uniform estimate $\sup R_\infty^\frac{q'-p}{q'p}=\max(\sqrt{R_\infty},1)$ where $p\leq q'\leq\infty$ $2\leq p<\infty$.
   For $1\leq p<\infty$ we would get instead $\max(R_\infty,1)$. 
   From the non--negativity of $R$ it follows that:
   \begin{equation}
      R_\infty
      =k(1-\underset{\nu\in U}{\einf} A_k(\nu))
   \end{equation} 
   From \eqref{eq:tfanalysis:crossamb:properties} it follows 
   that the inequality $R_\infty\leq 1$ is always fulfilled for C1.
   For the case C2 this gives instead that $R_\infty\leq 4$, in general.
\end{proof}
The following lemma provides a simple upper bound on $E_p/\lVert\BHSpread^{(\alpha)}_{\BH}\rVert_q$ 
which is for $p=2$ uniformly in $g$ and $\gamma$. Thus, it will serve as a benchmark.
\begin{mylemma}[Uniform Bound for $E_p/\lVert\BHSpread_{\BH}^{(\alpha)}\rVert_q$]
   For $1\leq p<\infty$ and
   $1<q\leq\infty$ it holds: 
   \begin{equation}   
      \frac{E_p(\mu)}{\lVert\BHSpread^{(\alpha)}_{\BH}\rVert_q}\leq
      \lVert V\rVert_\infty^\frac{p-2}{p}\cdot k^{2/p}\cdot |U|^{1/q'}
      \label{eq:approxeigen:ep:uniform}
   \end{equation}
   with $k=1$ for C1 and $k=2$ for C2. 
   \label{lemma:approxeigen:ep:uniform}
\end{mylemma}
\begin{proof}
   We use $R_1\leq R_\infty\cdot |U|$ in \eqref{eq:lemma:approxeigen:lemmaep3:eq1} of Lemma
   \ref{lemma:approxeigen:lemmaep3} and the uniform estimates $R_\infty\leq k$.
   from Lemma \ref{lemma:approxeigen:cpb:uniform}.
\end{proof}
This bound can not be related to ambiguity functions, i.e. will give no insight
on possible improvements due to localization.

\subsection{Weighted Norms of Ambiguity Functions and Localization}
\label{subsec:localization}
In the previous section we have shown that $R_1$ is a relevant term, which controls
the approximate eigenstructure. 
In the following analysis we shall further investigate $R_1$. We are interested in 
$\inf_{g,\gamma}(R_1)$ which is:
\begin{equation}
   \inf_{g,\gamma} R_1=k|U|\left(1-\sup_{g,\gamma}\langle A_k,\BHScat\rangle\right)
   \label{eq:approxeigen:localization:R1}
\end{equation}
where $\BHScat:=\chi_U/|U|$.
Thus, \eqref{eq:approxeigen:localization:R1} is a particular case
of a more general problem, where $\BHScat$ is some arbitrary weight (non--negative) function
$\BHScat$.
Thus, let us consider  $\sup_{g,\gamma}\langle A_k,\BHScat\rangle$ and
let us focus first only on $A_1=|\Amb_{g\gamma}^{(\alpha)}|^2$ which is also positive.
Since $A_1$ is quadratic in $\gamma$ we can rewrite 
$\langle A_1,\BHScat\rangle=\langle \gamma,L_{\BHScat,g}\gamma\rangle$ where this quadratic
form defines (weakly) an operator $L_{\BHScat,g}$. 
Such operators are also called \emph{localization operators} \cite{daubechiez:tflocalization:geophase}
and it follows  that $\sup_{\gamma}\langle A_1,\BHScat\rangle=\lambda_{\max}(L_{\BHScat,g})$.
The eigen--values and eigen--functions of Gaussian ($g$ is set to be a Gaussian)
localization operators on the disc ($U$ is a disc) 
are known to be Hermite functions (more generally this holds if $\BHScat$ has elliptical symmetry).
Kozek \cite{kozek:thesis,kozek:eigenstructure} found that
for elliptical symmetry also the joint optimization results in Hermite 
functions\footnote{Kozek considered $g=\gamma$. However one can show that
  for elliptical symmetry around the origin the optimum has also this property.}.
For $\BHScat$ being Gaussian 
the joint optimum ($g$ \emph{and} $\gamma$) is known explicitly \cite{jung:isit06}. 
The last result is based on
a theorem, formulated in \cite{jung:isit06}, which we will need
also in this paper. Let us consider for simplicity once again the 
one--dimensional case (the generalizations for $n>1$ are similar),
i.e. for  $n=1$ we have:
\begin{mytheorem}
   Let $\lVert g\rVert_2=\lVert\gamma\rVert_2=1$ and 
   $s,r\in\Reals$. Furthermore
   let $\BHScat\in\Leb{s'}(\Reals^2)$. Then the inequality:
   \begin{equation}
      \langle|\Amb^{(\alpha)}_{g\gamma}|^r, \BHScat \rangle \leq 
      \left(\frac{2}{rs}\right)^{\frac{1}{s}}\lVert\BHScat\rVert_{s'}
      \label{eq:jung:fidelitybound}
   \end{equation}
   holds for each $s\geq\max\{1,\frac{2}{r}\}$.
   \label{thm:jung:fidelitybound}
\end{mytheorem}
From \eqref{eq:weyl:shift:alphageneralized} follows that \eqref{eq:jung:fidelitybound} does not depend 
on the polarization $\alpha$.
The proof can be found in \cite{jung:isit06} and is based on a result of 
E. Lieb \cite{lieb:ambbound}. Note that apart from the 
normalization constraint the bound in Theorem \ref{thm:jung:fidelitybound} 
does not depend anymore on $g$ and $\gamma$. Hence 
for any given  $\BHScat$ the optimal bound $N_r(\BHScat)$
can be found by
\begin{equation}
   N_r(\BHScat):=\min_{\Reals\ni s\geq\max\{1,\frac{2}{r}\}}\left(
     \left(\frac{2}{rs}\right)^{\frac{1}{s}}\lVert\BHScat\rVert_{s'}
   \right)
   \label{eq:jung:fidelitybound:min}
\end{equation}
The equality case in Theorem \ref{thm:jung:fidelitybound} 
is given for $g$,$\gamma$ and $\BHScat$ being Gaussians (see \cite{jung:isit06}
for more details).
The following lemma states lower and upper bounds on the optimal achievable values of
the quantities $\langle A_k,\BHScat\rangle$.
\begin{mylemma}
   Let be $\BHScat:\Reals^{2n}\rightarrow\RealsPlus$ a non--negative weight function with $\lVert\BHScat\rVert_1=1$. 
   Then it holds:
   \begin{equation}
      \lambda_{\max}(Q^*Q)\leq \sup_{g,\gamma}\langle A_1,\BHScat\rangle\leq
      N_2(\BHScat)
   \end{equation}
   for case C1 and equivalently for case C2:
   \begin{equation}
      \lambda_{\max}(Q^*Q)^{1/2}=\max_{g,\gamma}\langle A_2,\BHScat\rangle\leq
      N_1(\BHScat)
   \end{equation}
   where $Q$ is the operator with spreading function $\BHScat$ in polarization $\alpha$.
   \label{lemma:approxeigen:localization}
\end{mylemma}
\begin{proof}
   Considering first the case C1 (that is $k=1$), which
   is independent of the polarization $\alpha$. The corresponding term
   $\langle A_1,\BHScat\rangle$ is relevant in the theory of WSSUS pulse shaping \cite{jung:wssuspulseshaping}
   where $\BHScat$ is called the scattering function.
   In \cite{jung:isit05} we have already pointed out that a lower bound can be obtained from
   convexity. We have:
   \begin{equation}
       |\langle g,Q\gamma\rangle|^2\leq\langle A_1,\BHScat\rangle\leq N_2(\BHScat)
   \end{equation}
   where $Q$ is a compact (follows from normalization) operator with spreading function $\BHScat$.
   The uniform upper bound is according to \eqref{eq:jung:fidelitybound:min}.
   The optimum of the lower bound is achieved for $g$ and $\gamma$ being the eigen--functions
   of $Q^*Q$ and $QQ^*$ corresponding to the maximal eigen--value $\lambda_{\max}(Q^*Q)$, such that
   for the supremum over $g$ and $\gamma$ it follows that:
   \begin{equation}
       \lambda_{\max}(Q^*Q)\leq\sup_{g,\gamma}\langle A_1,\BHScat\rangle\leq N_2(\BHScat)
   \end{equation}   
   For the case C2 ($k=2$) we proceed as follows.
   For a given $\gamma$ we have:
   \begin{equation}
      \langle A_2,\BHScat\rangle=\frac{1}{2}\left(\langle Q\gamma,g\rangle+\langle g,Q\gamma\rangle\right)
      \leq\lVert Q\gamma\rVert_2
   \end{equation}
   with equality in the last step for $g=Q\gamma/\lVert Q\gamma\rVert_2$. Choosing $\gamma$ 
   from the eigen--space of $Q^*Q$ related to the maximal eigen--value, we get:
   \begin{equation}
      \lambda_{\max}(Q^*Q)^{1/2}
      =\max_{g,\gamma}\langle A_2,\BHScat\rangle\leq N_1(\BHScat)
   \end{equation}
   because
   $\langle A_2,\BHScat\rangle
   \leq\langle|A_2|,\BHScat\rangle\leq\langle\sqrt{A_1},\BHScat\rangle
   \leq N_1(\BHScat)$ where again $N_1$ is from
   \eqref{eq:jung:fidelitybound:min}.
\end{proof}
For the particular weight function of interest in this paper, i.e. for $\BHScat=\chi_U/|U|$ the upper
bounds can be calculated explicitely. For $n=1$ we get the following result:
\begin{mycorollary}[Norm Bounds for Flat Scattering]
   Let be $\BHScat:=\chi_U/|U|$.
   Then it holds that:
   \begin{equation}
      \langle |\Amb^{(\alpha)}_{g\gamma}|^r, \BHScat \rangle <N_r(\BHScat)=
      \begin{cases}
         e^{-\frac{r|U|}{2e}} & |U|\leq 2e/r^*\\
         \left(\frac{2}{r^*|U|}\right)^{r/r^*} & \text{else}
      \end{cases}
      \label{eq:ambbound:supp:best}
   \end{equation}
   where $r^*=\max\{r,2\}$. It is not possible to achieve equality. 
   \label{corr:ambbound:supp}
\end{mycorollary}
The proof is obviously independent of $\alpha$ and available in \cite{jung:isit06}. 
\begin{myremark}
   When using the WSSUS model \cite{bello:wssus} for doubly--dispersive 
   mobile communication channels one typically assumes 
   time--frequency scattering within a shape
   $U=[0,\tau_d]\times[-B_d,B_d]$
   such that $|U|=2B_d\tau_d\ll 1< e$, where $B_d$ denotes maximum Doppler bandwidth $B_d$
   and $\tau_d$ is maximum delay spread.
   Then \eqref{eq:ambbound:supp:best} predicts for a $\Leb{1}$--normalized scattering function $\BHScat:=|U|^{-1}\chi_U$, that 
   the best (mean) correlation response ($r=2$) in using filter $g$ at the 
   receiver and $\gamma$ at the transmitter is bounded above by 
   $e^{-2B_d\tau_d/e}$. 
\end{myremark}

From the definition of $R_1$ in \eqref{eq:approxeigen:R1} and from  
\eqref{eq:ambbound:supp:best} of Corollary \ref{corr:ambbound:supp} 
we know
that for $|U|\leq k e$ we have the estimate:
\begin{equation}
   k|U|(1-e^{-\frac{|U|}{ke}})<\inf_{g,\gamma}(R_1)\leq k|U|(1-\lambda_{\max}(Q^*Q)^{1/k})
\end{equation}
which are implicit inequalities for $|U|$. 
The restriction $|U|\leq e$ for the lower bound can be removed if the second
alternative in \eqref{eq:ambbound:supp:best} of Corollary \ref{corr:ambbound:supp} is further
studied. However, for simplicity
we have considered only the first region which is suited to our application (small $|U|$).
In particular, 
with $R_1\leq R_\infty |U|$ we have
also $R_\infty\geq k(1-e^{-\frac{|U|}{k e}})$.
This proves also the assertion in \cite{jung:isit07}, i.e.
a necessary condition for $R_\infty\leq1$ is that $|U|\leq 2e\ln2$ . Furthermore
for $R_\infty\rightarrow k$ the size constraint on $U$ vanishes.

\subsection{Even Spreading Functions and Laguerre Integrals}
\label{subsec:approxeigen:laguerre}
Simple estimates for $\langle|\Amb_{g\gamma}^{(\alpha)}|^r,\BHScat\rangle$ (and therefore
also for $\lambda_{\max}(Q^*Q)$)
can be found if $\BHScat$ exhibits certain symmetries upon canonical transformations.
Let $T:\Reals^{2n}\rightarrow\Reals^{2n}$ be the transformation 
$T(\nu)=L\cdot\nu+c$ with a $2n\times 2n$ symplectic matrix\footnote{
  This means that $\sympl(L\mu,L\mu)=\sympl(\mu,\mu)$ for all
  $\mu$. In particular this means that $|\det(L)|=1$ such that the 
  measure $|U|$ is invariant under $L$.  
} $L$ and a phase space translation $c\in\Reals^{2n}$. It is well known that
$|\Amb_{g\gamma}|=|\Amb_{\tilde{g}\tilde{\gamma}}\circ T|$, where $\tilde{g}$ and $\tilde{\gamma}$ are 
related to $g$ and $\gamma$ by unitary transforms which depend on $T$. 
See for example \cite[Chapter 4]{folland:harmonics:phasespace}
for a review on metaplectic representation. We have then:
$\langle |\Amb_{g\gamma}|^r,\BHScat\rangle=\langle |\Amb_{\tilde{g}\tilde{\gamma}}|^r,\BHScat\circ T^{-1}\rangle$.
In particular this means, that we can always rotate, translate and 
(jointly) scale $\BHScat$ to simple prototype shapes. For example, elliptical
(rectangular) shapes can always be transformed to discs (squares) 
centered at the origin. Further symmetries can be exploited as shown exemplary in
the following lemma (for simplicity we consider only $n=1$):
\begin{mylemma}
   Let be $Q$ the operator with spreading function $\chi_U$.
   If the shape of $U$ has the symmetry $\chi_U(\mu)=\chi_U(-\mu)$ then for each $m\geq0$ 
   it holds that:
   \begin{equation}
      \lambda_{\max}(Q^*Q)\geq 
      \left(\frac{1}{|U|}\int_U l_m(\pi(|\mu|^2))d\mu\right)^2
   \end{equation}
   where $|\mu|^2=\mu_1^2+\mu_2^2$ and $l_m$ is the $m$th Laguerre function.
   \label{lemma:approxeigen:infR1:symmetric}
\end{mylemma}
\begin{proof}
   The calculation of $\lambda_{\max}(Q^*Q)$ simplifies much for normal operators which 
   involves the investigation of $Q$ only, i.e. $\lambda_{m}(Q^*Q)=|\lambda_{m}(Q)|^2$. 
   For an arbitrary operator $Y$ it follows
   that $\BHSpread_{Y^*}^{(\alpha)}(\mu)=\bar{\BHSpread}^{(\alpha)}_{Y}(-\mu)e^{-i4\pi\alpha\asy(\mu,\mu)}$
   is the spreading function of $Y^*$ in polarization $\alpha$.
   Hence, on the level of spreading functions the normality of $Y$ is equivalent to:
   \begin{equation}
      \BHSpread_{Y}^{(\alpha)}(\mu)\bar{\BHSpread}_Y^{(\alpha)}(\nu)=
      \bar{\BHSpread}_{Y}^{(\alpha)}(-\mu)\BHSpread_{Y}^{(\alpha)}(-\nu)\cdot
      e^{i4\pi\alpha(\asy(\mu,\mu)+\asy(\nu,\nu))}
   \end{equation}
   which can be verified using the rules for $\Shift_\mu^{(\alpha)}$ like \eqref{eq:weyl:shift:alphageneralized}
   and \eqref{eq:weyl:shift:weylcommrelation}. 
   The operator $Q$ has by definition the real spreading function $\chi_U$. 
   Hence the desired symmetry is fulfilled  for $\alpha=0$.
   Let be $h_m$ the $m$th Hermite function. 
   It is known that the ambiguity functions of Hermite functions are given by the Laguerre functions 
   \cite{klauder:radardesign} (see for example also \cite{folland:harmonics:phasespace}).
   Obviously, the maximal eigen--value fulfills:
   \begin{equation}
      \lambda_{\max}(Q)\geq\langle h_m,Qh_m\rangle
      =\frac{1}{|U|}\int_U\langle h_m,\Shift_\mu^{(0)}h_m\rangle d\mu
      =\frac{1}{|U|}\int_U l_m(\pi|\mu|^2)d\mu
      \label{eq:corr:approxeigen:lambdan}
   \end{equation}
   where $l_m(t)=e^{-t/2}L_m^{(0)}(t)$ are the Laguerre functions and
   $L_m^{(0)}$ are the $0$th Laguerre polynomials. 
\end{proof}

\subsection{Gaussian Signaling and the Corresponding Bounds}
\label{subsec:approxeigen:gaussiansignaling}
\newcommand{\cOne}{e}

The previous part of this section indicates that approximate eigen--functions have to be ''Gaussian--like''. 
Hence it makes sense to consider Gaussian signaling explicitely. For simplicity we do this
for the time--frequency symmetric case $g=\gamma=2^{\frac{1}{4}}e^{-\pi t^2}$ and $n=1$. 
We have the relation:
\begin{equation}
   (\Shift_\nu^{(\alpha)} g)(x) 
   =e^{-\pi\left(2i\sympl(\nu,x\cOne)+i(1-2\alpha)\asy(\nu,\nu)+\nu_1^2\right)}
   g(x)
   \label{eq:approxeigen:numerical:gausshift}
\end{equation}
if we let  $\cOne:=(1,i)$.
According to \eqref{eq:approxeigen:lemmaep1:proof1} the error $E_p(\mu)$ can be calculated as:
\begin{equation}
   \begin{split}
      E_p(\mu)
      =\lVert\int_{\Reals^2} d\nu\, \BHSpread_{\BH}^{(\alpha)}(\nu)e^{-i2\pi\sympl(\nu,\mu)}
      \cdot g\cdot f(\nu,\cdot)\rVert_p\\      
   \end{split}
   \label{eq:approxeigen:gaussian:Ep}
\end{equation}
The function $f:\Reals^2\times\Reals\rightarrow\Complexes$ is defined 
for a particular polarization $\alpha$ as:
\begin{equation}
   \begin{split}
      f(\nu,x)
      &=\begin{cases}
         e^{-\pi[2i\sympl(\nu,x\cOne)+i(1-2\alpha)\asy(\nu,\nu)+\nu_1^2]}-
         e^{-\pi[\frac{\langle\nu,\nu\rangle}{2}-2i\alpha\asy(\nu,\nu)]}
         & \text{for C1}\\
         e^{-\pi[2i\sympl(\nu,x\cOne)+i(1-2\alpha)\asy(\nu,\nu)+\nu_1^2]}-1
         & \text{for C2}
      \end{cases}
   \end{split}
   \label{eq:approxeigen:numberical:fnux}
\end{equation}
where we have used that the ambiguity function in polarization $\alpha$ is
$\Amb_{gg}^{(\alpha)}(\nu)=e^{-\frac{\pi}{2}s_\alpha(\nu)}$ and
$s_\alpha(\nu):=\DotReal{\nu}{\nu}+4i\alpha\asy(\nu,\nu)$.
The following Corollary contains the bounds specialized to the Gaussian case:
\begin{mycorollary}[Gaussian Bounds]
   For the case C1 ($k=1$) and for the case C2 ($k=2$) in polarization $\alpha=0$  it holds
   for any $1\leq p<\infty$ and $1\leq q\leq\infty$ that:
   \begin{equation}
      \frac{E_p(\mu)}{\lVert\BHSpread^{(\alpha)}_{\BH}\rVert_{q}}
      \leq
      32^{\frac{p-2}{4p}}\cdot \lVert k (1-e^{-\frac{\pi}{k}s_0})\chi_U\rVert^{1/p}_{q'/p}
      \label{eq:corr:approxeigen:gaussianbounds:eq1}
   \end{equation}
   where $s_0(\nu):=\langle\nu,\nu\rangle$. For $q>1$ it follows
   from \eqref{eq:corr:approxeigen:gaussianbounds:eq1} also:
   \begin{equation}
      \begin{split}
         \frac{E_p(\mu)}{\lVert\BHSpread^{(\alpha)}_{\BH}\rVert_{q}}
         \leq
         32^{\frac{p-2}{4p}}\cdot k\cdot
         \left(k|U|(1-
           \langle\BHScat, e^{-\frac{\pi}{k}s_0}\rangle\right)^{1/\max(q',p)}
      \end{split}
      \label{eq:corr:approxeigen:gaussianbounds:eq2}
   \end{equation}
   where $\BHScat=\chi_U/|U|$.
   \label{corr:approxeigen:gaussianbounds}
\end{mycorollary}
\begin{proof}   
   We use the abbreviation $A_1=e^{-\pi s_0}$ and  $A_2=\Real{e^{-\frac{\pi}{2}s_\alpha}}$
   as introduced in \eqref{eq:approxeigen:Ak}.
   Only for $\alpha=0$ the case C2 provides an Euclidean distance measure
   in phase space. 
   Equation
   \eqref{eq:corr:approxeigen:gaussianbounds:eq1} of the claim follows
   from Lemma \ref{lemma:approxeigen:lemmaep1} and from \eqref{eq:lemma:approxeigen:lemmaep3:eq1} of
   Lemma \ref{lemma:approxeigen:lemmaep3} together with    
   $\lVert V\rVert_\infty\leq 2\lVert g\rVert_\infty=32^{1/4}$.
   If $q>1$ we can relate this further by 
   \eqref{eq:lemma:approxeigen:lemmaep3:eq2} of
   Lemma \ref{lemma:approxeigen:lemmaep3}
   to weighted norms of ambiguity functions. Using the uniform 
   bound $C_{pq}\leq k$ from Lemma 
   \ref{lemma:approxeigen:cpb:uniform} and the relation for $R_1$
   in \eqref{eq:approxeigen:R1:special} we get Gaussian integrals of the form
   \eqref{eq:corr:approxeigen:gaussianbounds:eq2}
   which can now be solved analytically for some cases. 
   For example, if
   $U$ is a centered disc of radius $\sqrt{|U|/\pi}$ we
   get $\langle \BHScat,e^{\frac{\pi}{k}s_0}\rangle=l(2|U|/k)$ where $l(x)=2(1-e^{-x/2})/x$.
   For a centered square of length $\sqrt{|U|}$ we have instead
   $l(x)=2\text{erf}(\sqrt{\pi x/8})^2/x$.
\end{proof}
\section{Numerical Verification}
\label{sec:numerical}
In this part we shall establish a spreading model with a finite number of random parameters.
We shall need this model to verify numerically the bounds derived in this paper. 
Since several (iterated)
integrals are involved which partially can only be computed numerically we have 
evaluate the achieved accuracy. We aim at computing 
$E_p(\mu)/\lVert\BHSpread_{\BH}^{(\alpha)}\rVert_q$ up to a desired accuracy $\Delta$. 
In our derivation we
will assume that single definite integrals can be computed within a given 
predefined error (for example in using Simpson quadrature).

\subsection{Spreading Model with Finite Number of Parameters}
Let us consider a doubly--dispersive channel model with a finite number 
of fading parameters $c_k$, where $k\in\setZ_K^2$ and 
$\setZ_K=\{0\dots K-1\}$. Each fading contribution has its
\emph{own doubly--dispersive} operation on the input signal, hence the model
is different from the usual (distributional) models 
having a finite number of separated paths with fixed Doppler frequencies.
The spreading function $\BHSpread_{\BH}^{(\alpha)}$ should be of the form:
\begin{equation}
   \BHSpread_{\BH}^{(\alpha)}(\nu)=\sum_{k\in\setZ_K^2}c_k\chi_u(\nu-u(k+o))
   =\sum_{k\in\setZ_K^2}c_k\chi_1(\nu/u-k+o)
   \label{eq:approxeigen:finitespreadingmodel}
\end{equation}
where $\chi_u(y)=\chi_{[0,u]}(y_1)\chi_{[0,u]}(y_2)$ is the characteristic function
of the square $[0,u]\times[0,u]=:[0,u]^2$ and $o=(\tfrac{1}{2},\tfrac{1}{2})$. 
Thus the latter is a disjoint partition of the 
square $[0,Ku]^2$ with area $(Ku)^2$. In other words,
if we fix the support of the spreading function to be $|U|$, then it follows for a $K^2$--sampling
of this area that $u=\sqrt{|U|}/K$. For such a model the $q$--norm of the spreading function 
as needed for the calculation of the ratio $E_p/\lVert\BHSpread_{\BH}^{(\alpha)}\rVert_q$
is:
$\lVert\BHSpread_{\BH}^{(\alpha)}\rVert_q=u^{2/q}\lVert c\rVert_q$
where $\lVert c\rVert_q:=(\sum_k |c_k|^q)^{1/q}$ is simply the $q$th vector norm of the vector 
$c=(\dots,c_k,\dots)\in\Complexes^{K^2}$.
Let us abbreviate $l=l(k)=k+o$. 
With \eqref{eq:approxeigen:numberical:fnux} we get for the integrand in
\eqref{eq:approxeigen:gaussian:Ep}:
\begin{equation}
   \begin{split}
      \int_{\Reals^2} \BHSpread_{\BH}^{(\alpha)}(\nu)
      e^{-i2\pi\sympl(\nu,\mu)}g(x)f(\nu,x) d\nu
      &=\sum_{k\in\setZ_K^2} c_k\cdot g(x) 
      \underbrace{\int_{\Reals^2}\chi_1(\frac{\nu}{u}-l)e^{-i2\pi\sympl(\nu,\mu)}f(\nu,x) d\nu}_{F_k(x)}\\
   \end{split}
   \label{eq:approxeigen:numerical:Fk:1}
\end{equation}
The approximate eigenstructure error reads now as $E_p(\mu)=\lVert\sum_{k\in\setZ_K^2} c_k \cdot g \cdot F_k\rVert_p$.
For $\alpha=1/2$ and case C2 the integral in $F_k(x)$ can be calculated explicitely. In general, however, 
$F_k(x)$ has to be computed numerically up to a certain accuracy $\delta$ (it is a well--defined and
definite integral). Thus, let
the computed value $\tilde{F}_k(x)$ be such that pointwise
$|\tilde{F}_k(x)-F_k(x)|\leq\delta$ for all $x$ and $k$.
We would like to use $\tilde{F}_k(x)$ instead of $F_k(x)$ to compute 
the approximation $\tilde{E}_p(\mu)$ on $E_p(\mu)$. However we have 
to restrict the remaining indefinite integral over $x$ to a finite interval
$I:=[-L,L]$. With $J$ we denote its complement in $\Reals$, i.e. $J:=\Reals\without I$.
Observe from \eqref{eq:approxeigen:numberical:fnux} that $|f|\leq 2$, hence $|F_k|\leq 2u^2$
in \eqref{eq:approxeigen:numerical:Fk:1} and that for a Gaussian $\lVert g\rVert_p^p=1/\sqrt{p}$.
If we choose $\pi L\geq \max(\sqrt{\log(2u^2/\delta)},1)$ we have:
\newcommand{\erfc}{\text{erfc}}
\begin{equation}
   \begin{split}
      \lVert F_k g\cdot\chi_J\rVert_p
      &\leq 2u^2 \lVert g\cdot\chi_J\rVert_p
      =2u^2\erfc(\sqrt{\pi p}L)^{1/p}
      \leq \frac{2u^2}{\left(\pi\sqrt{p}L\right)^{1/p}}e^{-\pi L^2}\\
      &=\lVert g\rVert_p \frac{2u^2}{\left(\pi L\right)^{1/p}}e^{-\pi L^2}
      \overset{\pi L\geq 1}{\leq}
      \lVert g\rVert_p \cdot 2u^2 \cdot e^{-\pi L^2}\leq \delta\lVert g\rVert_p
   \end{split}
\end{equation}
For such a chosen $L$ the integration with respect to $x$ over the interval $I=[-L,L]$ can be performed again within an
accuracy of $\delta$. This yields for the overall calculation error:
\begin{equation}
   \begin{split}
      |E_p(\mu)-\tilde{E}_p(\mu)|
      &\leq \delta+
      \sum_{k} |c_k|\left(
        \lVert(F_k-F_k^{(\delta)})g\cdot\chi_I\rVert_p+
        \lVert F_k g\cdot\chi_J\rVert_p
      \right)\\
      &\leq \left(1+
        2\lVert c\rVert_1\cdot\lVert g\rVert_p\right)\delta
      = \left(1+
        2\lVert c\rVert_1\cdot p^{-\frac{1}{2p}}\right)\delta
   \end{split}
\label{eq:approxeigen:num:Ep}
\end{equation}
If we  choose  $\delta=\Delta\cdot\lVert\BHSpread^{(\alpha)}_{\BH}\rVert_q\cdot(1+2\lVert c\rVert_1\cdot p^{-\frac{1}{2p}})^{-1}$ (and $L$ respectively) we
can guarantee that 
the error on $E_p(\mu)/\lVert\BHSpread^{(\alpha)}_{\BH}\rVert_q$ is below $\Delta$.
\begin{myremark}[Interference Estimates for Statistical Models]
   \label{rem:approxeigen:statisticalmodel}
   Consider the following example: The transmitter sends the signal $\Shift_\mu\gamma$ through the unknown
   channel $\BH$. Let us again for simplicity use
   the finite--parameter spreading model \eqref{eq:approxeigen:finitespreadingmodel} 
   for a support $U$ of square shape. The receiver already knows 
   the vector of fading parameters $c$ for the spreading function $\BHSpread^{(\alpha)}_{\BH}$ of the channel,
   the pulse $g$ and $\gamma$ and the time--frequency slot $\mu$. 
   The normalized $q$--norms $c_q=\lVert c\rVert_q\cdot K^{-2/q}$ 
   of the  $K^2$ fading coefficients characterize the statistical model for
   the spreading such that $\lVert\BHSpread^{(\alpha)}_{\BH}\rVert_q=|U|^{1/q}\cdot c_q$.
   If the contribution of this particular slot $\mu$ is removed from the signal it remains
   $e:=\BH\Shift_\mu\gamma-\lambda(\mu)\Shift_\mu g$. Let us assume that the receiver expects 
   another information in the span of the function $f$ (for example $f=\Shift_\nu g$ could be another
   slot $\nu$). The interference will be $\langle f,e\rangle$.
   Let be $A_f(p)=\lVert f\rVert_{p'}\cdot 32^\frac{p-2}{4p}$.
   We have
   \begin{equation}
      |\langle f,e\rangle|\leq E_p(\mu)\cdot \lVert f\rVert_{p'}
      <A_f(p)\cdot(|U|(1-L^2))^{1/\max(q',p)}\cdot |U|^{1/q}\cdot c_q
   \end{equation}
   With the assumption that $|U|\leq1$ we use $|U|^{1/\max(q',p)+1/q}\leq |U|$ such that:
   \begin{equation}
      |\langle f,e\rangle|
      <A_f(p)\cdot(1-L^2)^{1/\max(q',p)}\cdot |U|\cdot c_q
   \end{equation}
   This means, for different statistical models (characterized by $c_q$) and functions $f$ 
   (characterized by $\lVert f\Vert_{p'}$ in the quantity $A_f(p)$) we can 
   characterize the amount of interference. 
\end{myremark}
\subsection{Numerical Experiments}
We will consider now the case where the coefficients $c_k$ of the 
vector $c\in\Complexes^{K^2}$ are identical, independent and normal distributed which
refers to the doubly--dispersive Rayleigh fading channel. The 
square shaped support $U$ has a random size $|U|$ taken from a
distribution uniformly on the interval $[10^{-3}, 10^{-2}]$ corresponding
to values of the time--frequency spread relevant in mobile 
communication.
Each realization of the fading factor $c$ and $u=\sqrt{|U|}$ parameterize via \eqref{eq:approxeigen:finitespreadingmodel}
a random spreading function $\BHSpread_{\BH}^{(\alpha)}$ in a given polarization $\alpha$
which give itself rise to a random channel operator $\BH$ by
Lemma \ref{lemma:tfanalysis:linop:spreading}. 
On this random channel we
have evaluated $E_p(\mu)$ for Gaussian signaling as described previously
in Section \ref{subsec:approxeigen:gaussiansignaling}. 
For each realization we have taken $\mu$ uniformly from $[-5,5]^2$.
We have calculated $N=1000$ Monte Carlo (MC) runs for different values of $p$ and $q$. 
For each run $E_p(\mu)/\lVert\BHSpread_{\BH}^{(\alpha)}\rVert_q$ has been computed (corresponding
to one point in Fig.\ref{fig:approxeigen:ep:C1:2:2} and Fig.\ref{fig:approxeigen:ep:C1:3:15}) up to an accuracy
of $\Delta=10^{-8}$.
The computed
values $E_p(\mu)$ are compared to the uniform bound in
\eqref{eq:approxeigen:ep:uniform} of Lemma \ref{lemma:approxeigen:ep:uniform} 
which depends only on the support and is valid for any normalized $g$ and $\gamma$.
Improved bounds are valid only for particular $g$ and $\gamma$ like
the Laguerre/Gauss (GL) bound from Theorem \eqref{thm:approxeigen:Ep3}.
Fig.\ref{fig:approxeigen:ep:C1:2:2} shows the case C1 for $p=q=2$, where we expect
the most tight results. 
The  GL bound improves the uniform estimates approximately by a factor of 10. However
the computed MC values are still below this estimate by a factor of approximately two. The 
latter estimate degrades to a factor of approximately $10$ for $p=3$ and $q=3/2$ 
as displayed in Fig.\ref{fig:approxeigen:ep:C1:3:15}.

\begin{figure}
   \includegraphics[width=\linewidth]{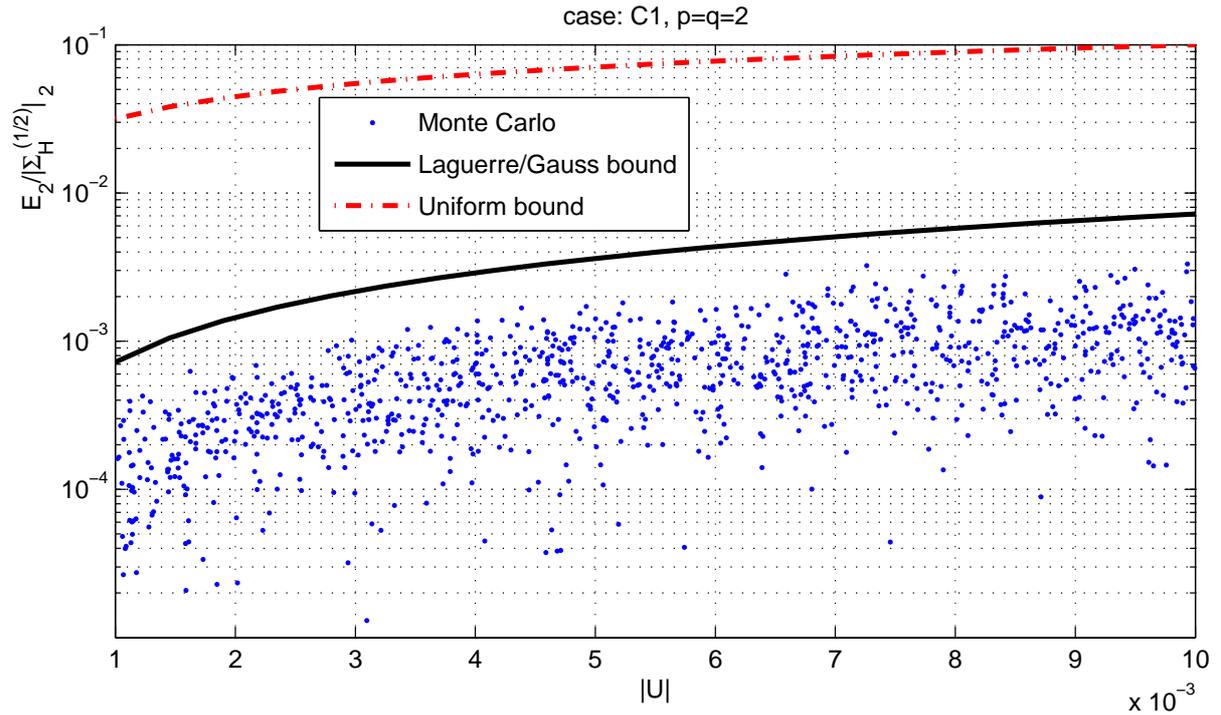}
   \caption{\emph{Approximate Eigenstructure for the case C1, $p=2$, $q=2$:} Verification of $1000$
     Monte Carlo runs with the uniform bound in Lemma \ref{lemma:approxeigen:ep:uniform} and the 
     optimized Laguerre/Gauss bound  of Theorem \ref{thm:approxeigen:Ep3}.     
   }
   \label{fig:approxeigen:ep:C1:2:2}
\end{figure}

\begin{figure}
   \includegraphics[width=\linewidth]{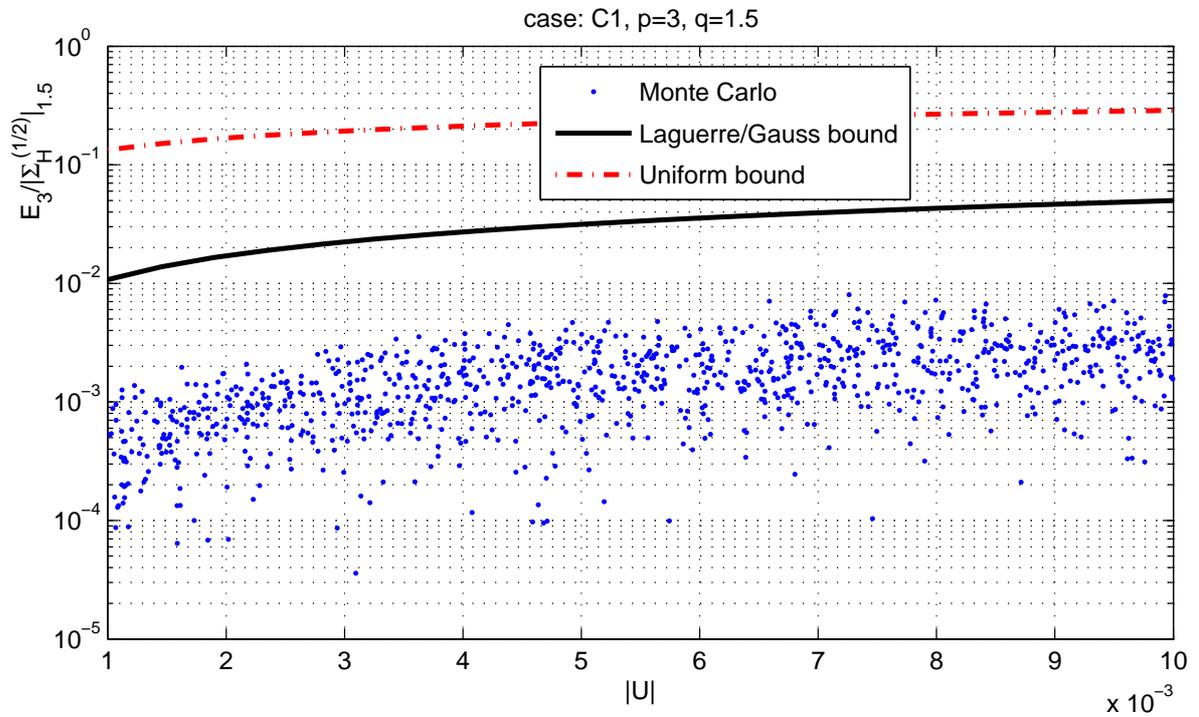}
   \caption{\emph{Approximate Eigenstructure for the case C1, $p=3$, $q=1.5$:} Verification of $1000$
     Monte Carlo runs with the uniform bound in Lemma \ref{lemma:approxeigen:ep:uniform} and the 
     optimized Laguerre/Gauss bound  of Theorem \ref{thm:approxeigen:Ep3}.     
   }
   \label{fig:approxeigen:ep:C1:3:15}
\end{figure}
%

%
\section{Conclusions}
In this paper we have considered doubly--dispersive channels with compactly supported spreading. 
We have shown to what level of approximation error a description
as simple multiplication operators is valid. 
We have focused on two well known choices of such a description, i.e. the
multiplication with the (generalized) Weyl symbol of the operator and the case
of Wigner smoothing. 
We found that in both cases the approximation errors can be 
bounded by the size of the support of the spreading function. Our estimates improve 
and generalize recent results in this direction. Furthermore we have drawn the relation to 
localization operators and fidelity measures known from the theory of pulse shaping.
Finally, we have verified our estimates using Monte Carlo methods with a precise
control of the numerical uncertainties.
\section*{Acknowledgments}
\addcontentsline{toc}{section}{Acknowledgment}
The author would like to thank Gerhard Wunder, 
Thomas Strohmer and Holger Boche for helpful discussions on this topic. 
The author extents a special thanks to the anonymous reviewers for their
constructive comments, which have improved the paper.

%
%
\appendix

\subsection{Proof of Lemma \ref{lemma:gmc:underspread:generalupperbound}}
\label{appendix:lemma:gmc:underspread:generalupperbound:proof}
The following proof is motivated by \cite{kozek:thesis}.\\
\begin{myproof}
   For each complex Hilbert space with $\lVert x\rVert_2^2=\langle x,x\rangle$
   the following inequality 
   \begin{equation}
      \begin{split}
         \lVert x-y\rVert_2^2
         &\leq\underbrace{\lVert x\rVert_2^2-\lVert y\rVert_2^2}_{\text{(a)}}+
         2\underbrace{|\langle y,x-y\rangle|}_{\text{(b)}} \end{split}
   \end{equation}
   holds. Now let $x=\BH\Shift_\mu^{(\alpha)}\gamma$ 
   and $y=\BHWeyl^{(\alpha)}_{\BH}(\mu)\Shift_\mu^{(\alpha)}\gamma$.
   Using \eqref{eq:tfanalysis:noncomm:shiftproperty} the following 
   upper bounds
   \begin{equation}
      \begin{split}
         \text{(a)} 
         &=   \langle\gamma,\left(\Shift_\mu^{(\alpha)*}\BH^*\BH\Shift_\mu^{(\alpha)}-
           \BHWeyl^{(\alpha)}_{\BH^*\BH}+\BHWeyl^{(\alpha)}_{\BH^*\BH}-|\BHWeyl^{(\alpha)}_{\BH}|^2\right)\gamma\rangle\\
         &\leq|\BHWeyl^{(\alpha)}_{\BH^*\BH}-|\BHWeyl^{(\alpha)}_{\BH}|^2|+
         |\langle\gamma,\left(\Shift_\mu^{(\alpha)*}\BH^*\BH\Shift_\mu^{(\alpha)}-
           \BHWeyl^{(\alpha)}_{\BH^*\BH}\right)\gamma\rangle|\\
         &=|\BHWeyl^{(\alpha)}_{\BH^*\BH}-|\BHWeyl^{(\alpha)}_{\BH}|^2|+|\int_{\Reals^{2n}}\BHSpread^{(\alpha)}_{\BH^*\BH}(\nu)e^{-i2\pi\sympl(\nu,\mu)}(\Amb^{(\alpha)}_{\gamma\gamma}(\nu)-1)d\nu\,|\\
         &\leq|\BHWeyl^{(\alpha)}_{\BH^*\BH}-|\BHWeyl^{(\alpha)}_{\BH}|^2|+\lVert\BHSpread^{(\alpha)}_{\BH^*\BH}\Omega\rVert_1\\
         \text{(b)} 
         &=|\BHWeyl^{(\alpha)}_{\BH}(\mu)|\cdot
         |\langle\gamma,\left(\Shift_\mu^{(\alpha)*}\BH\Shift_\mu^{(\alpha)}-\BHWeyl^{(\alpha)}_{\BH}\right)\gamma\rangle|\\
         &=|\BHWeyl^{(\alpha)}_{\BH}(\mu)|\cdot|\int_{\Reals^{2n}}\BHSpread^{(\alpha)}_{\BH}(\nu)e^{-i2\pi\sympl(\nu,\mu)}(\Amb^{(\alpha)}_{\gamma\gamma}(\nu)-1) d\nu\,|
         \leq|\BHWeyl^{(\alpha)}_{\BH}(\mu)|\cdot\lVert\BHSpread^{(\alpha)}_{\BH}\Omega\rVert_1
      \end{split}
   \end{equation}
   will give the proposition.
\end{myproof}

\subsection{Proof of Lemma \ref{lemma:approxeigen:lemmaep3}}
\label{appendix:lemma:approxeigen:lemmaep3:proof}
\begin{myproof}  
   Firstly -- note that H\"older's inequality for the index pair $(1,\infty)$ gives
   $V^p_p\leq V_{\infty}^{p-2}\cdot V_2^{2}$
   with equality for $p=2$; and in turn $\lVert V_p\rVert_{q'}\leq
   \lVert  V_{\infty}^{\frac{p-2}{p}}\cdot V_2^{\frac{2}{p}}\rVert_{q'}$. For $q>1$ we can 
   rewrite this and use again H\"olders inequality. We get:
   \begin{equation}
      \lVert V_p\rVert_{q'}\overset{q>1}{=}
      \lVert  V_{\infty}^{\frac{q'(p-2)}{p}}\cdot V_2^{\frac{2q'}{p}}\rVert_1^{1/q'}
      \leq \lVert V\rVert_\infty^\frac{p-2}{p}\cdot\lVert V_2^2\rVert_{q'/p}^{1/p}
      \overset{\eqref{eq:approxeigen:defR}}{=} \lVert V\rVert_\infty^\frac{p-2}{p}\cdot R^{1/p}_{q'/p}
      \label{eq:approxeigen:lemmaep3:proof03}
   \end{equation}
   For $q=1$ (i.e. $q'=\infty$) we obtain rhs of the last equation directly:
   \begin{equation}
      \lVert V_p\rVert_\infty=
      \lVert  V_{\infty}^{\frac{p-2}{p}}\cdot V_2^{\frac{2}{p}}\rVert_\infty
      \leq \lVert V\rVert_\infty^\frac{p-2}{p}\cdot R^{1/p}_{\infty}
   \end{equation}
   which proves \eqref{eq:lemma:approxeigen:lemmaep3:eq1} of this lemma. 
   From the definition of $R$ in \eqref{eq:approxeigen:defR} it is obvious that the minimum of the bounds 
   is taken at $B(U)=\Amb^{(\alpha)}_{g\gamma}(U)$ which is provided by C1.
   Because there it holds always equality for $p=2$ this is also the optimizer for 
   $\lVert V_2\rVert_{q'}$ for any $q$. 
   From \eqref{eq:approxeigen:lemmaep3:proof03} 
   we get further for $p\leq q' <\infty$:
   \begin{equation}      
      R^{1/p}_{q'/p}=R_\infty^{1/p}\lVert R/R_\infty\cdot \chi_U\rVert^{1/p}_{q'/p}
      \leq R_\infty^\frac{q'-p}{q'p} R^{1/q'}_{1}
      \label{eq:approxeigen:lemmaep3:proof04}
   \end{equation}
   because in this case $(R(\nu)/R_\infty)^{q'/p}\leq R(\nu)/R_\infty$ for all $\nu\in U$. 
   For $q'=p$ equality occurs in the last inequality. 
   This proves \eqref{eq:lemma:approxeigen:lemmaep3:eq1} of this lemma
   for $q'\geq p$.
   For $q'<p$ we use the concavity of $R^{q'/p}$, i.e. we proceed instead as follows:
   \begin{equation}           
      \begin{split}
         R^{1/p}_{q'/p}
         &=
         \left(|U|\cdot\lVert R^{q'/p}\chi_U/|U|\,\rVert_1\right)^{1/q'}\\
         &\leq
         |U|^{1/q'}\cdot\lVert R\chi_U/|U|\,\rVert_1^{1/p}
         \leq |U|^\frac{p-q'}{q'p}\cdot R_1^{1/p}
      \end{split}
      \label{eq:approxeigen:lemmaep3:proof05}
   \end{equation}
   The bounds \eqref{eq:approxeigen:lemmaep3:proof04} and \eqref{eq:approxeigen:lemmaep3:proof05} agree
   for $q'=p$ and are tight for $q=p=2$. 
\end{myproof}

\bibliographystyle{IEEEtran}
\bibliography{references}
\end{document}


%% file: notations.tex
\usepackage{graphics,graphicx,color}

\catcode`\@=11
\@ifpackageloaded{beamer}{
}{
}
\catcode`\@=12

\newcommand{\OpClass}{\text{OP}}

\newcommand{\twistconv}{\,\natural\,}

\newcommand{\DotReal}[2]{\langle#1,#2\rangle}

\newcommand{\asy}{\zeta}

\newcommand{\sympl}{\eta}

\newcommand{\Trace}[1]{\,\mathbf{Tr}{#1}}
\newcommand{\setZ}{{\mathbb{Z}}}

\newcommand{\schattenclass}{\mathcal{T}}
\newcommand{\settraceclass}{\schattenclass_1}
\newcommand{\without}{\setminus}

\newcommand{\defeq}{\overset{\text{def}}{=}}

\newcommand{\HH}{{\mathcal{H}}}
\newcommand{\BH}{\boldsymbol{\HH}}

\newcommand{\BHSpread}{\boldsymbol{\Sigma}}
\newcommand{\BHWeyl}{\boldsymbol{L}}
\newcommand{\BHScat}{\boldsymbol{C}}

\newcommand{\Leb}[1]{\mathcal{L}_{#1}}
\newcommand{\Ltwo}{\Leb{2}}

\newcommand{\Amb}{{\mathbf{A}}}

\newcommand{\Real}[1]{{\text{Re}}\{#1\}}

\newcommand{\Shift}{{\boldsymbol{S}}}

\newcommand{\Schwarz}{{\mathcal{S}}}

\newcommand{\Sinc}{\text{sinc}}

\newcommand{\Fourier}{\mathcal{F}}
\newcommand{\sFourier}{\mathcal{F}_s}
\newcommand{\nFourier}{\mathbf{F}}

\newcommand{\support}[1]{\text{\rm supp}\,(#1)}

\newcommand{\esup}{\text{\rm ess sup\,\,}}
\newcommand{\einf}{\text{\rm ess inf\,\,}}

\newcommand{\Reals}{\mathbb{R}}
\newcommand{\RealsPlus}{\mathbb{R}_{+}}
\newcommand{\Complexes}{\mathbb{C}}

\newcommand{\taumax}{{\tau_d}}
